\documentclass[12pt]{article}
\usepackage{amsfonts,amssymb,stmaryrd,frf}
\begin{document}
\title{On orbifolds and free fermion constructions}
\author{Ron Donagi \and Katrin Wendland}
\maketitle
\begin{abstract}
This work develops the correspondence between orbifolds and free fermion models.
A complete classification is obtained for orbifolds $X/G$ with 
$X$ the product of three elliptic curves and $G$ an abelian
extension of a group $(\Z_2)^2$ of twists acting on $X$. 
Each such quotient $X/G$ is shown to give a geometric interpretation
to an appropriate free fermion model, including the geometric NAHE+ model.
However, the semi-realistic NAHE free fermion model is proved to be 
non-geometric: its Hodge numbers are not reproduced by any orbifold $X/G$.
In particular cases it is shown that $X/G$ can agree with some
Borcea-Voisin threefolds, an orbifold limit of the Schoen threefold,
and several further orbifolds thereof. This yields free fermion models with geometric interpretations on such special threefolds.
\end{abstract}
\section*{Introduction}

This work explores a class of heterotic string theories, more precisely of heterotic
conformal field theories, and their geometric interpretations. We
consider quantum field theories that arise by means of so-called free 
fermion constructions, and we study the geometric counterparts of the 
resulting models. Free fermion models are interesting in this context, because
mathematically, they are comparatively simple. They all yield rational
conformal field theories, which makes them mathematically well behaved. 
On the other hand, there are free fermion models which can be interpreted
as nonlinear sigma models on tori. In other words, there are special points
in the moduli space of conformal field theories on tori, where the 
corresponding conformal field theories allow a free fermion construction. Hence
for some particular models, there are geometric interpretations
at hand, and the notion of ``geometric interpretation'' can indeed be made
mathematically precise. Finally, more general free fermion models can be included
into the discussion by implementing orbifold techniques. 

This raises the natural question whether one can find free fermion models
which on the one hand yield semi-realistic string theories, in that they 
produce exactly the spectrum of the minimal supersymmetric standard model
in the observable massless sector, and which on the other hand allow a
geometric interpretation on a geometric orbifold of a torus. In other words, do 
any free fermion models exist
which connect both to the real world, via the standard model of particle
physics, and to geometry, via a mathematically tractable geometric 
interpretation?

Addressing the first part of the task, 
to our knowledge,  \cite{muwi86} contains the first hint that free 
fermion models could be used to construct 
semi-realistic models by orbifold-like procedures. 
These ideas have been further developed by many authors, and
models with semi-realistic gauge groups are given e.g.\
in \cite{fny90,inq87}, see also \cite{cfn99}.
Further references on free fermion models are
\cite{klt87,gool85,abkw86, abk87, ab88}, and the reader interested in the related
topic of covariant lattice approaches could consult 
\cite{fms86,cfq86,lls86,lelu87,bfhv87,lls87,lns87,luth88,ltz88}.

An example of a model of interest to us in this context is the so--called NAHE model 
\cite{fgkp87,aehn89,fana93}. 
It is an example of a semi-realistic heterotic string theory,
and it can be obtained from a toroidal model by a chain of orbifoldings
of type $\Z_2$. In fact, a closer study reveals that a geometric 
interpretation on an orbifold of a torus, if it exists, must have the
form $X/G$ with $X$ the product of three elliptic curves and 
$G$ a semidirect product of a group $G_S$ of shifts on $X$ and
a subgroup $G_T\subset G$ which is isomorphic to $(\Z_2)^2$, see
\cite{fa93,fft06}.

One is hence naturally led to a classification problem: To determine
all topologically inequivalent Calabi-Yau threefolds that arise by
resolving the quotient singularities in $X/G$ with $X$ the product of
three elliptic curves and $G$ a group of the type described above. This 
problem is solved in the present paper. A partial classification was
already given in \cite{dofa04}, under additional restrictions on
the ``group of twists'' $G_T$. A classification of Calabi-Yau threefolds
$X/G$ for which $G_T$ is isomorphic to $(\Z_n)^2$
with $n\neq2$ was given by Jimmy Dillies in \cite{di07}. From these
classifications one finds a negative answer to the question posed above:
No purely geometric interpretation 
of the semi-realistic NAHE free fermion model exists, since 
for none of the groups $G$ described above,
the Hodge numbers $h^{1,1},\, h^{2,1}$ of the resolution of $X/G$ yield 
three generations  $h^{1,1}-h^{2,1}=3$. In other words, the NAHE and other
semi-realistic free fermion models must involve some non-geometric 
orbifolds.

The goal of this paper is to solve the geometric classification problem,
to embed it into the context of free fermion models, and to point out some
interesting geometric and model-building features arising from the classification. 

We start in section 1 with the classification of quotients $X/G$ with $G$ 
as described above. We give a complete list, including the Hodge numbers 
of the resulting resolved Calabi-Yau threefolds, as well as their fundamental 
groups. We also include an (incomplete) discussion of possible coincidences 
within our list.

In section 2 we show that for each of the Calabi-Yau threefolds 
in our list there exists a free fermion model 
whose underlying geometry is $X/G$. We start with a mathematical
review of free fermion constructions. We state and explain the rules of the game,
and we discuss orbifolds in the free fermion language. We rederive the well-known
fact that a particular free fermion model allows a geometric interpretation on an
$SO(12)$ torus. This, along with the discussion of orbifolds, allows us to show
that indeed for each model in our list of orbifolds $X/G$, there is an associated
free fermion model. 

Our list includes a number of Calabi-Yau threefolds 
that are familiar from other contexts. 
The simplest of these is
the Vafa-Witten threefold
$X/(\Z_2)^2$ studied in \cite{vawi95}.
The NAHE+ model, capturing the geometric part of the NAHE model, is another example.
Contrary to popular lore, it is \emph{NOT} a $\Z_2$ orbifold of the Vafa-Witten threefold. 
We show instead that it can be obtained as a $\Z_2 \times \Z_2$ orbifold of the Vafa-Witten 
threefold. The full NAHE model is not geometric: we do not obtain any three-generation models in our classification. 
For other examples, we recover six different types of Borcea-Voisin threefolds
\cite{bo97,vo93} within our list. We also
find orbifold limits of Schoen's threefold \cite{sch88} and some of its
orbifolds within our list of quotients $X/G$. This may be of considerable 
interest because precisely these threefolds have been successfully 
used in the construction of semi-realistic heterotic string theories in
\cite{dopw00,bodo06,bcd06,bodo07}. If an appropriate degenerate limit
of the relevant gauge bundles can be found, then our result will lead
to a dramatic simplification of these heterotic constructions: Free fermion
models, after all, are mathematically well understood and technically
easy to handle.

Discrete torsion may be included in our orbifolds without leaving the realm 
of free fermion constructions. We note that turning on discrete torsion has 
a rather mild effect on the Hodge numbers of our threefolds: we get many of 
the Hodge numbers of models without torsion, and the only new Hodge pairs are 
mirrors of existing pairs. Similar observations in more specialized situations 
have been made before, e.g. in \cite{dw00,prrv07}.
According to Vafa and Witten \cite{vawi95}, full mirror 
symmetry (as opposed to just the Hodge theoretic matching) 
is indeed sometimes realized through discrete torsion. 
This situation may be specific to $(\Z_2)^2$ orbifolds though, as suggested
in \cite{ks95}. 
The conclusion of \cite{ks94}, suggesting that asymmetric orbifolds should 
be related to discrete torsion, applies in a different setting, where
the emphasis lies on simple current constructions but not on geometric 
interpretations. 
The NAHE model is not obtainable as a geometric orbifold, with or without discrete torsion.
Among the six Borcea-Voisin threefolds we obtain, three are their own mirrors, while the other three are exceptional in the sense that they do not have mirrors within the 
Borcea-Voisin construction. Our result that for these
threefolds, there exist associated free fermion models, could therefore well be 
useful to shed some light on aspects of mirror symmetry and discrete 
torsion for these threefolds.

Our basic classification is 
accomplished with the help of some simple reduction principles, which reduce 
the combinatorial complexity and allow 
us to do everything by hand. Without these reductions, the amount of calculations 
required is massive. Indeed, several computer searches have been carried out 
recently on regions in the string landscape that overlap ours to various degrees. 
Nooij \cite{cfn03,no04} studied $\Z_2$-type free fermion models 
based on the $SO(12)$ torus. He includes 
non-geometric orbifolds, and finds a handful of three generation models. A partial list of orbifolds and Hodge numbers is obtained in \cite{prrv07}. In work in progress, these authors are studying orbifolds with generalized discrete torsion. This apparently leads them to recover precisely the complete list of Hodge numbers obtained here. The coincidence is quite intriguing; it would be interesting to know whether the objects themselves coincide or whether the Hodge numbers simply fail to capture the relevant data. Note that, for example,
the fundamental groups of their models have not been computed.
In another work in progress, Kiritisis, Lennek and Schellekens 
\cite{kls08} are searching certain free fermion models whose
partition functions are left-right symmetric. Due to a language barrier, 
it is difficult to compare their models directly to ours. The list of Hodge numbers 
they get apparently agrees with ours, except that they get one additional model, 
with Hodge numbers (25,1). The latter is clearly not geometric in our sense:
By our assumptions on the orbifolding group $G$, the $G$-invariant part of
$H^\ast(X,\R)$ contains three dimensional subspaces of $H^{1,1}(X,\R)$ and
of $H^{1,2}(X,\R)$, respectively. Hence 
the Hodge numbers of all our geometric orbifolds arise by adding contributions of 
various twisted sectors to the basic (3,3) contribution of the bulk sector, so our Hodge numbers must be at least 3.\\[5pt]

{\textbf{Acknowledgements.}}

The research leading to this paper has been performed at various locations. The research of R.D. has been supported by NSF grants DMS 0139799 and DMS 0612992, and by Research and Training Grant DMS 0636606. K.W. cordially thanks CIRM at Luminy, France and the Penn Math/Physics Group for their hospitality. Her repeated visits to Philadelphia have been partly funded by Penn's NSF Focused Research Grant, DMS 0139799, and by her Nuffield Award to Newly Appointed Lecturers in Science, Engineering and Mathematics, NAL/00755/G. We have 
benefitted from discussions with A. Bak, V. Bouchard, J. Dillies, A. Faraggi, E. Kiritsis, 
M. Kreuzer, M. Ratz, and B. Schellekens.

\section{A classification of relevant orbifolds}\label{classification}
In this section, we discuss a classification of orbifoldings and orbifolds.
Restricting to groups whose so-called twist group $G_T$ is
isomorphic to $(\Z_2)^2$, we introduce a notion of equivalence among such
groups, via a reduction principle. Orbifolding the product
of three elliptic curves by one group yields a quotient which is isomorphic to what is obtained from the product of three different (but isogenous) elliptic curves by an equivalent group.
We give a classification of all such groups up to equivalence.
We also calculate some topological data of the resulting orbifolds, namely their 
Hodge numbers and their fundamental groups. This gives further information
about possible isomorphies among the respective quotients.
The main results are the tabulation of orbifolds in Section \ref{tables} and the somewhat incomplete analysis of coincidences in Section \ref{coincidences}.
\subsection{On a classification of toroidal orbifolds}\label{classtor}
We work with a 6 (real) dimensional torus $X \cong T^6$ with the complex
structure of a product $E_1 \times E_2 \times E_3$ of three elliptic curves. Let
$T_0 \cong (\Z_2)^2 \subset (\Z_2)^3$
be the Klein group of twists acting on
$(z_1,z_2,z_3) \in X$
by an even number of sign changes:
\begin{eqnarray*}
t_1:&& (z_1,z_2,z_3) \to (z_1,-z_2,-z_3), \\
t_2:&& (z_1,z_2,z_3) \to (-z_1,z_2,-z_3), \\
t_3:&& (z_1,z_2,z_3) \to (-z_1,-z_2,z_3). 
\end{eqnarray*}
An arbitrary automorphism $g$ of X can be factored uniquely:
$g = s \circ g_t$,
where the twist part $g_t$ is an automorphism sending the origin $0 \in X$
to itself, while the shift part $s$ is translation by 
$g(0) \in X$. Any group $G$ of automorphisms fits in an exact sequence
$$ 
0 \longrightarrow G_S \longrightarrow G \stackrel{\pi}{\longrightarrow} G_T^0 \longrightarrow 0,
$$
where $G_S$ is the subgroup of shifts contained in $G$, and $G_T^0$ is the
group of twist parts
of all elements of $G$, so
$G_T^0 := \left\{g_t | \exists\right.$ a shift $s$ such that 
$\left.g = s \circ g_t \in G\right\}$. 
In general, $G_T^0$ is not a subgroup of $G$. However, it  follows
from Lemma \ref{essential} below that we can always reduce to a situation where we can
choose a subgroup $G_T\subset G$ 
which maps isomorphically onto $G_T^0$ under $\pi$, and such that $G=G_S\times G_T$.

Our goal in this section is to study toroidal orbifolds, i.e.\ quotients
$X/G$, for all finite groups $G$ whose twist part is $T_0$. We will see
that these come in a finite number of irreducible families.
\begin{definition}
We say that a group $G$ of automorphisms of $X$ is {\em
redundant} if it contains a translation by a non zero $x \in E_i$ for some
$i \in {1,2,3}$, and is {\em essential} otherwise.
\end{definition}
Our first observation (cf.\cite{dofa04}) is that there is a simple reduction
principle: every toroidal orbifold $X/G$ with $X = E_1 \times E_2 \times E_3$ and a
given twist part $G_T^0$ is also of the form $X^\prime/G^\prime$ for some 
$X^\prime = E_1^\prime \times E_2^\prime \times E_3^\prime$ and some essential group of automorphisms 
$G^\prime$ with the same
twist part $G_T^0$. Indeed, if the redundant $G$ contains a translation
$s_x$ by a non zero element $x \in E_i$, then $x$ must be a torsion
element, the quotient $E^\prime_i := E_i / x$ is an elliptic curve, the quotient
$X^\prime  := X/x$ is still a product of three elliptic curves with one $E_i$
replaced by $E^\prime_i$, and
$$
X/G \cong X^\prime/G^\prime , 
$$
where $G^\prime := G / \langle s_x\rangle$ fits into an exact sequence
$$ 
0 \to G^\prime_S \to G^\prime \to G^\prime_T \to 0,
$$
with $G^\prime_S = G_S / \langle s_x\rangle$ and $G^\prime_T = G_T^0$ as claimed.

We therefore may as well restrict attention to essential groups $G$.
\begin{lemma}\label{essential} 
Any essential group $G$ with twist part $G^0_T = T_0$ is
commutative and isomorphic to the direct product $G_S \times G^0_T$ of
its shift and twist parts. All elements of $G$ are of order $2$, and 
up to conjugation $G$ is
contained in $G^{max}$ which is the extension
$$ 
0 \to X[2] \to G^{max} \to T_0 \to 0, 
$$
where $X[2] \cong (\Z_2)^6$ is the group of all points of order $2$ in $X$.
\end{lemma}
\begin{proof} 
First we show that any $g \in G$ has order 2. Let
$g=s \circ g_t \in G$ with $s\in G_S$ a shift by $x\in X$
and $g_t\in G_T^0=T_0$. If $g_t \neq 1 \in G^0_T$ then $g^2$ is a shift
along $x+g_t(x)$, i.e.\ along
one of the three elliptic curves, so essence implies $g^2=1$.  We
still need to consider $g \in G_S$. The subgroup $G_S$ of $G$ is isomorphic to
$$
\overline{G}_S := \left\{x \in X | s_x \in G_S\right\}.
$$
The latter is invariant under the action of $T_0$. If it contains
$x=(x_1,x_2,x_3)$ it must also contain $t_i(x)$ hence $x+t_i(x)$, which is
in $E_i$. Essence therefore implies that $2x_i = 0$ for all $i$.

It follows that $G$ is commutative and contains a subgroup $G_T$ that maps isomorphically onto the twist group $G^0_T$. Further, it follows that $G$ is isomorphic to the direct product $G_S \times G_T$ of its shift subgroup $G_S$ with any such $G_T$. Now $G_S$ is a group of translations by points of order 2, so it is contained in $G^{max}$. The twist group $G_T$ need not be contained in $G^{max}$. Its generators can be written in the form:
\begin{eqnarray*}
(z_1,z_2,z_3) \to (x_1+z_1,x_2-z_2,x_3-z_3), \\
(z_1,z_2,z_3) \to (y_1-z_1,y_2+z_2,y_3-z_3)
\end{eqnarray*}
The order-2 condition requires that $x_1, y_2$ and $y_3-x_3$ be points of order 2, while the three remaining variables are unconstrained in the three elliptic curves $E_i$. Nevertheless, one checks immediately that conjugation by an appropriate translation of $X$ (which also has three complex degrees of freedom, one in each $E_i$) can be chosen to set $x_2=x_3=y_1=0$. Such a conjugation takes $G_T$ into $G^{max}$ and leaves $G_S$ unchanged, completing the proof.

\vspace{3pt}
\end{proof}

In view of the lemma, our essential group $G$ contains a ``subgroup of twists'' $G_T$ which 
under $\pi$ maps
isomorphically to $G^0_T$, and $G$ is isomorphic to $G_S \times G_T$. In
the next section we will see that up to conjugation there are four
possible actions of the twist group $G_T$ on $X$.
\subsection{Classification of essential automorphism groups}
\begin{definition} 
The {\em rank} of an essential automorphism group $G$ is the
rank of $G_S$ as a module over $\Z_2$.
\end{definition}
We will study the possible automorphism groups according to their
increasing rank. We will usually describe an automorphism group in terms
of a minimal set of generators, listing each generator in the form of a triple
$(\epsilon_1 \delta_1, \epsilon_2 \delta_2, \epsilon_3 \delta_3)$, where
$\epsilon_i \in E_i$ is a point of order 2, and $\delta_i \in \{ \pm \} $
indicates the pure twist part. We take the period lattice of the elliptic
curve $E_i$ to be generated by $2$ and $2\tau$, so the $\epsilon_i$ can
be one of $0,1,\tau, 1+\tau$. The three operations that produce equivalent
groups are change of basis, permutation of the three coordinates $z_i$ of
the torus, and a shift of one or more of the $z_i$. 
We start with rank $0$, where instead of listing two generators
we often list all three non zero group elements.
\begin{lemma}\label{rankzero}
There are 4 inequivalent groups $G=G_T$ of rank $0$, given as
follows:
\begin{eqnarray*}
(0-1):&& (0+,0-,0-),(0-,0+,0-),(0-,0-,0+),\\
(0-2):&& (0+,0-,0-),(0-,0+,1-),(0-,0-,1+),\\
(0-3):&& (0+,0-,0-),(0-,1+,1-),(0-,1-,1+),\\
(0-4):&& (1+,0-,0-),(0-,1+,1-),(1-,1-,1+).
\end{eqnarray*}
\end{lemma}
\textsl{Remark}: 
In \cite{dofa04}, only the first of these possibilities, as well as its
further quotients, were considered, leading to the considerably shorter
list there.
\\[3pt]
\begin{proof} 
Any rank 0 group is generated by two elements of the form
$(\epsilon_1 +, \epsilon_2 -, \epsilon_3 -)$ and $(\epsilon_4 -,
\epsilon_5 +, \epsilon_6 -)$. By shifting the three coordinates $z_i$ we
can clearly arrange that $\epsilon_2=\epsilon_3=\epsilon_4=0$, and by
changing the labeling of a homology basis for the $E_i$ we can take each
of the remaining $\epsilon_i$ to be $0$ or $1$. This leaves us with 8
possibilities, including the four above and
\begin{eqnarray*}
&&(0+,0-,0-),(0-,1+,0-),(0-,1-,0+);\\
&&(1+,0-,0-),(0-,0+,0-),(1-,0-,0+);\\
&&(1+,0-,0-),(0-,0+,1-),(1-,0-,1+);\\
&&(1+,0-,0-),(0-,1+,0-),(1-,1-,0+).
\end{eqnarray*}
Of these, the first two are equivalent to $(0-2)$ under a permutation of the
three coordinates. The third is transformed by a shift of $z_3$ to
$(1+,0-,1-)$, $(0-,0+,0-)$, $(1-,0-,1+)$ which is equivalent to $(0-3)$ under a
permutation of $z_1,z_2$. Similarly, the fourth group is transformed by a
shift of $z_1$ to $(1+,0-,0-)$, $(1-,1+,0-)$, $(0-,1-,0+)$ which under a
permutation of $z_2,z_3$ is equivalent to the third group, hence to $(0-3)$.

One can use similar elementary means to check that the four groups in the
statement of the lemma are inequivalent. In Section \ref{tables} we will find
the stronger result that the corresponding quotients $X/G$ are
topologically inequivalent.
\vspace{3pt}
\end{proof}
For a group $G$ of higher rank, we list first two generators of $G$ which
map onto a minimal generating set for the twist group $G_T$ in the previous $(\epsilon_1
\delta_1, \epsilon_2 \delta_2, \epsilon_3 \delta_3)$ notation; the
remaining generators are chosen to be in the shift subgroup $G_S
\subset G$. Since in this case all the $\delta_i$ are $0$, we can omit
them, using instead the abbreviated notation $(\epsilon_1, \epsilon_2,
\epsilon_3)$.
\begin{prop} 
There are $11$ equivalence classes of essential groups in rank
$1$, $14$ in rank $2$, $6$ in rank $3$, one in rank $4$, and none in higher ranks.
They are listed in the first two columns of Table $1$, see Section \mb{\rm\ref{tables}}.
\end{prop}
The proof is elementary and somewhat tedious, using the tools introduced
in the proof of Lemma \ref{rankzero}. We leave the details to the reader.
\subsection{Orbifold cohomology}\label{orbcoh}
Though the techniques are well known, 
let us briefly summarize for the reader's convenience the procedure by which one calculates the Hodge numbers of a minimal resolution of the orbifold $X/G$. 
We will assume the situation which is of interest below, that is
$X=E_1\times E_2\times E_3$
equipped with the complex structure of the product of three elliptic curves $E_i$, and
that $G$ is an essential group of the type described in Section \ref{classtor}.
In particular, by Lemma \ref{essential},  $G=G_S\times G_T$ with $G_T\cong T_0$ under the
projection $\pi$ to the twist parts.

First observe that the cohomology of $X$ is obtained by taking the wedge product between
the total cohomologies of each elliptic curve $E_i$.
With respect to a local complex coordinate $z_i$ on $E_i$, the cohomology of the latter
is generated by $1,\, dz_i,\, d\overline z_i, \, dz_i\wedge d\overline z_i$. 
If $g\in G$ splits as $g=s\circ t_i$ with $t_i\in T_0$ into its shift and its twist part, then
$g$ acts on $dz_1,\, dz_2,\,dz_3$ by $dz_i\mapsto dz_i$ and $dz_j\mapsto -dz_j$ for $j\neq i$,
and similarly for the $d\overline z_k$. 
Hence for the $G$-invariant part of the cohomology of $X$ we find dimensions 
$h^{p,q}_{inv}$, $p,q\in\{0,\ldots,3\}$, with
$$
\left( h^{p,q}_{inv} \right)_{p,q} =
\left( \begin{array}{cccc}
1&0&0&1\\
0&3&3&0\\
0&3&3&0\\
1&0&0&1
\end{array}\right).
$$
For example, we have representatives $dz_i\wedge d\overline z_i$ in $H^{1,1}(X)$ and
$dz_1\wedge dz_2\wedge d\overline z_3$ in $H^{2,1}(X)$.

Additional contributions to the cohomology of the minimal resolution of $X/G$ come
from the blow-ups of curves of singularities. Assume that $g\in G,\, g\neq 1,$  
has fixed points on $X$. Since by assumption $g=s\circ t_i$ for
some  $t_i\in T_0$ and $s$ a shift, this implies that $s$ is a shift by some
point $x=(x_1,x_2,x_3)\in X$ of order $2$ with $x_i=0$. The fixed locus of $g$
thus consists of $16$ copies of $E_i$. In $X/G$, the image yields a curve of
singularities of type $A_1$. Its contributions to the cohomology of a resolution of
$X/\langle g\rangle$ have dimensions 
$h^{p,q}_{g}$, $p,q\in\{0,\ldots,3\}$, with
$$
\left( h^{p,q}_{g} \right)_{p,q} =
\left( \begin{array}{cccc}
0&0&0&0\\
0&16&16&0\\
0&16&16&0\\
0&0&0&0
\end{array}\right).
$$
The contributions to the cohomology of the resolved quotient $X/G$ are given
by the $G$-invariant part of these vector spaces. If $G$ has rank $r$, i.e.\ 
$G\cong (\Z_2)^{r+2}$, then the total contribution from blowing up
the fixed locus of $g=s\circ t_i$ is 
$$
\left( h^{p,q}_{g, inv, A} \right)_{p,q} =
\left( \begin{array}{cccc}
0&0&0&0\\
0&2^{3-r}&2^{3-r}&0\\
0&2^{3-r}&2^{3-r}&0\\
0&0&0&0
\end{array}\right)
\;\mb{ or }\;
\left( h^{p,q}_{g, inv, B} \right)_{p,q} =
\left( \begin{array}{cccc}
0&0&0&0\\
0&2^{4-r}&0&0\\
0&0&2^{4-r}&0\\
0&0&0&0
\end{array}\right).
$$
The case $h^{p,q}_{g, inv, B}$ applies if and only if the subgroup of $G$
which maps an irreducible component of the fixed locus of $g$ in $X$
onto itself is strictly larger than $\langle g\rangle$. Indeed, then $G$ contains  elements $h$
which map each copy of $E_i$
in the fixed locus of $g$ onto itself, but which act by multiplication by $-1$ on $dz_i$
and $d\overline z_i$, thus leaving none of the cohomology classes counted by $h^{2,1}_g$
and $h^{1,2}_g$ invariant, whereas all contributions to $h^{1,1}_g$ and
$h^{2,2}_g$ are invariant.

\subsection{Discrete torsion}\label{dt}

In his seminal paper \cite{va86}, Cumrun Vafa pointed out that in conformal field
theory, there is an additional degree of freedom $\eps\in H^2(G,U(1))$
when orbifolding by a group $G$, which is now commonly known as ``discrete torsion''.
Roughly speaking, one introduces a twisted action of $G$ on the contribution 
to the cohomology which comes from the blow-up of the singular locus in $X/G$. 
In the examples that are of interest for us, $G\cong (\Z_2)^{r+2}$. One 
checks that discrete torsion is compatible with the reduction principle of Section \ref{classtor},
and $H^2(G,U(1))=(\Z_2)^m$ with $m={r+2\choose 2}$. 
Consider elements $g,\,h\in G-\{1\}$  such that $h\neq g$ and $h$ maps each component
of the fixed locus of $g$ onto itself. Then the effect
of non-trivial discrete torsion $\eps(g,h)$
amounts to replacing the contributions $h^{p,q}_{g, inv, B}$ listed 
above by
$$
\left( h^{p,q}_{g, inv, \widetilde B} \right)_{p,q} =
\left( \begin{array}{cccc}
0&0&0&0\\
0&0&2^{4-r}&0\\
0&2^{4-r}&0&0\\
0&0&0&0
\end{array}\right).
$$
\subsection{Fundamental groups}
There is a simple procedure for calculating the fundamental group of an
orbifold, which goes back to \cite{dhvw85} in the physics literature. A
mathematical version can be found in \cite{brhi02}.

Let a group $\widetilde{G}$ act discretely on a simply connected $\widetilde X$. Let $F$
be the subgroup of $\widetilde{G}$ generated by all elements which have a
fixed point in $\widetilde X$. Then the fundamental group of the quotient space
$\widetilde X/\widetilde{G}$ is $\widetilde{G}/F$.

In our applications, we are interested in the fundamental group of 
quotients of the product $X$ of three elliptic curves, that is $X=\C^3/\Lambda$. We
take $\widetilde{G}$ to be the extension of the
orbifolding group $G$ by the lattice $\Lambda$:
$$  
0 \to \Lambda \to \widetilde{G} \to G \to 0,   
$$
so the orbifold is $X/G =
{\C}^3 /\widetilde{G}$. The calculation of the fundamental group of each
of our orbifolds is then a straightforward exercise.
\subsection{Tabulation of results}\label{tables}
\subsubsection*{Table 1: The list of automorphism groups}
We list the automorphism groups by rank. For each group $G$ we
list its twist group $G_T$, its shift part $G_S$ (if non-empty), the Hodge numbers
$h^{1,1},h^{2,1}$ of a small resolution of $X/G$, the fundamental group
$\pi_1(X/G)$, and the list of contributing sectors and their contribution.
For the fundamental groups we use the abbreviations:
\begin{eqnarray*}
A:&&        \mb{the extension of } \Z_2 \mb{ by } \Z^2 \mb{ (so } H_1(X) = (\Z_2)^3\mb{)}\\
B:&&        \mb{any extension of } (\Z_2)^2 \mb{ by } \Z^6 \mb{ (with various possible } H_1(X)\mb{)}\\
C:&&        \Z_2\\
D:&&        (\Z_2)^2
\end{eqnarray*}
A shift element is denoted by a triple $(\epsilon_1, \epsilon_2,
\epsilon_3)$, where $\epsilon_i \in E_i$ is a point of order 2,
abbreviated as one of $0,1,\tau, \tau1:=1+\tau$. A twist element is denoted by a
triple $(\epsilon_1 \delta_1, \epsilon_2 \delta_2, \epsilon_3 \delta_3)$,
where $\epsilon_i \in E_i$ is as above and $\delta_i \in \{ \pm \} $
indicates the pure twist part. A two-entry contribution $(a,b)$ adds $a$
units to $h^{1,1}$ and $b$ units to $h^{2,1}$. When $b=0$ we abbreviate
$(a,b)$ to the single entry contribution $a$. 
\\[1em]
$\begin{array}{|l|llc|c|c|}
\hline
&G_T&      G_S\hspace*{\fill}&      &         (h^{1,1},h^{2,1})&\pi_1\\ &
 &   \hspace*{\fill}\mb{sectors}  &       \mb{contribution}&         &\\ 
\hline\hline
\mb{{\textbf{Rank 0:}}}&&&&&\\\hline
(0-1)&(0+,0-,0-),(0-,0+,0-)&&&                                (51,3) &0\\ 
&&                0+,0-,0-    &    16&&\\ &
&                0-,0+,0-   &     16&&\\ &
&                0-,0-,0+   &     16&&\\
\hline
(0-2)&(0+,0-,0-),(0-,0+,1-)        &&&                        (19,19) &0\\ &
&                0+,0-,0-     &   8,8&&\\ &
&                0-,0+,1-      &  8,8&&\\
\hline
(0-3)&(0+,0-,0-),(0-,1+,1-)       &&&                         (11,11)&A\\ &
&                0+,0-,0-   &     8,8&&\\ 
\hline
(0-4)&(1+,0-,0-),(0-,1+,1-)       &&&                         (3,3)&B\\ 
\hline
\hline
\mb{{\textbf{Rank 1:}}}&&&&&\\\hline
(1-1)&(0+,0-,0-),(0-,0+,0-)   &     (\tau,\tau,\tau)    &&                    (27,3)&C\\ &
&                0+,0-,0-     &   8&&\\ &
&                0-,0+,0-     &   8&&\\ &
&                0-,0-,0+     &   8&&\\
\hline
(1-2)&(0+,0-,0-),(0-,0+,\tau-)  & (\tau,\tau,\tau)             &&    (15,15)           &     0\\ &
&                0+,0-,0-   &     4,4&&\\ &
&                0-,0+,\tau-    &    4,4&&\\ &
&                \tau-,\tau-,0+     &   4,4&&\\
\hline
(1-3)&(0+,0-,0-),(0-,0+,1-)  & (\tau,\tau,\tau)             &&    (11,11)              &  C\\ &
&                0+,0-,0-  &      4,4&&\\ &
&                0-,0+,1-   &     4,4&&\\
\hline
(1-4)&(0+,0-,0-),(0-,1+,1-)  & (\tau,\tau,\tau)             &&    (7,7)    &            A\\ &
&                0+,0-,0-  &      4,4&&\\
\hline
(1-5)&(1+,0-,0-),(0-,1+,1-)  & (\tau,\tau,\tau) &&                (3,3)       &         B\\
\hline
(1-6)&(0+,0-,0-),(0-,0+,0-)   &(\tau,\tau,0)            &&     (31,7)      &          0\\ &
&                0+,0-,0- &       8&&\\ &
&                0-,0+,0-  &      8&&\\ &
&                0-,0-,0+   &     8&&\\ &
&                \tau-,\tau-,0+    &    4,4&&\\
\hline
(1-7)&(0+,0-,0-),(0-,0+,1-) &  (\tau,\tau,0)              &&   (11,11)         &       C\\ &
&                0+,0-,0-   &     4,4&&\\ &
&                0-,0+,1-    &    4,4&&\\
\hline

\end{array}$\\
$\begin{array}{|l|llc|c|c|}
\hline
&G_T&      G_S\hspace*{\fill}&      &         (h^{1,1},h^{2,1})&\pi_1\\ &
 &   \hspace*{\fill}\mb{sectors}  &       \mb{contribution}&         &\\ 
\hline\hline

(1-8)&(0+,0-,0-),(0-,1+,0-)  & (\tau,\tau,0)               &&  (15,15)       &         0\\ &
&                0+,0-,0-  &      4,4&&\\ &
&                0-,1-,0+   &     4,4&&\\ &
&                \tau-,\tau1-,0+   &    4,4&&\\
\hline

(1-9)&(0+,0-,0-),(0-,1+,1-) &  (\tau,\tau,0)              &&   (7,7)      &          A\\ &
&                0+,0-,0-   &     4,4&&\\
\hline

(1-10)&(1+,0-,0-),(0-,1+,0-)  & (\tau,\tau,0)           &&      (11,11)  &              A\\ &
&                1-,1-,0+     &   4,4&&\\ &
&                \tau1-,\tau1-,0+    &  4,4&&\\
\hline

(1-11)&(1+,0-,0-),(0-,1+,1-)  & (\tau,\tau,0)            &&     (3,3)           &     B\\ 
\hline

\hline
\mb{{\textbf{Rank 2:}}}&&&&&\\\hline\hline
(2-1)&(0+,0-,0-),(0-,0+,0-)  & (1,1,1),(\tau,\tau,\tau)      &&   (15,3)   &             D\\ &
&                0+,0-,0-      &  4&&\\ &
&                0-,0+,0-     &   4&&\\ &
&                0-,0-,0+    &    4&&\\
\hline
(2-2)&(0+,0-,0-),(0-,0+,1-)  & (1,1,1),(\tau,\tau,\tau)    &&     (9,9)           &     C\\ &
&                0+,0-,0-   &     2,2&&\\ &
&                0-,0+,1-  &      2,2&&\\ &
&                1-,1-,0+ &       2,2&&\\
\hline
(2-3)&(0+,0-,0-),(0-,0+,0-)  & (1,1,1),(\tau,\tau,0)     &&    (17,5)            &    C\\ &
&                0+,0-,0-    &    4&&\\ &
&                0-,0+,0-   &     4&&\\ &
&                0-,0-,0+  &      4&&\\ &
&                \tau-,\tau-,0+ &       2,2&&\\
\hline
(2-4)&(0+,0-,0-),(0-,0+,1-) &  (1,1,1),(\tau,\tau,0)   &&      (11,11)       &         0\\ &
&                0+,0-,0-      &  2,2&&\\ &
&                0-,0+,1-     &   2,2&&\\ &
&                1-,1-,0+    &    2,2&&\\ &
&                \tau1-,\tau1-,0+ &     2,2&&\\
\hline
(2-5)&(0+,0-,0-),(0-,0+,\tau-)  & (1,1,1),(\tau,\tau,0)    &&     (7,7)           &     D\\ &
&                0+,0-,0-  &      2,2&&\\ &
&                0-,0+,\tau- &       2,2&&\\
\hline
(2-6)&(0+,0-,0-),(0-,0+,0-)  & (1,1,1),(\tau,1,0)      &&   (19,7)       &         0\\ &
&                0+,0-,0-     &   4&&\\ &
&                0-,0+,0-    &    4&&\\ &
&                0-,0-,0+   &     4&&\\ &
&                \tau1-,0+,1- &      2,2&&\\ &
&                \tau-,1-,0+ &       2,2&&\\
\hline
(2-7)&(0+,0-,0-),(0-,0+,\tau-) &  (1,1,1),(\tau,1,0)     &&    (9,9)       &         C\\ &
&                0+,0-,0-     &   2,2&&\\ &
&                0-,0+,\tau-    &    2,2&&\\ &
&                \tau1-,0+,\tau1- &     2,2&&\\
\hline
\end{array}$\\
$\begin{array}{|l|lrc|c|c|}
\hline
&G_T&      G_S\hspace*{\fill}&      &         (h^{1,1},h^{2,1})&\pi_1\\ &
 &   \hspace*{\fill}\mb{sectors}  &       \mb{contribution}&         &\\
\hline
\hline
(2-8)&(0+,0-,0-),(0-,\tau+,\tau-)  & (1,1,1),(\tau,1,0)   &&      (5,5)       &         A\\ &
&                0+,0-,0-       & 2,2&&\\
\hline

(2-9)&(0+,0-,0-),(0-,0+,0-)  & (0,1,1),(1,0,1)      &&   (27,3)          &      0\\ &
&                0+,0-,0-      &  4&&\\ &
&                0-,0+,0-     &   4&&\\ &
&                0-,0-,0+    &    4&&\\ &
&                0+,1-,1-   &     4&&\\ &
&                1-,0+,1-  &      4&&\\ &
&                1-,1-,0+ &       4&&\\
\hline
(2-10)&(0+,0-,0-),(0-,0+,\tau-) &  (0,1,1),(1,0,1)       &&  (11,11)          &      0\\ &
&                0+,0-,0-      &  2,2&&\\ &
&                0+,1-,1-     &   2,2&&\\ &
&                0-,0+,\tau-    &    2,2&&\\ &
&                1-,0+,\tau1-  &     2,2&&\\
\hline
(2-11)&(0+,0-,0-),(0-,\tau+,\tau-) &  (0,1,1),(1,0,1)     &&    (7,7)         &       A\\ &
&                0+,0-,0-  &      2,2&&\\ &
&                0+,1-,1- &       2,2&&\\
\hline
(2-12)&(\tau+,0-,0-),(0-,\tau+,\tau-) &  (0,1,1),(1,0,1)   &&      (3,3)        &        B\\
\hline
(2-13)&(0+,0-,0-),(0-,0+,0-)  & (1,1,0),(\tau,\tau,0)   &&      (21,9)        &        0\\ &
&               0+,0-,0- &       4&&\\ &
&                0-,0+,0-       & 4&&\\ &
&                0-,0-,0+      &  4&&\\ &
&                1-,1-,0+     &   2,2&&\\ &
&                \tau-,\tau-,0+    &    2,2&&\\ &
&                \tau1-,\tau1-,0+ &     2,2&&\\
\hline
(2-14)&(0+,0-,0-),(0-,0+,1-) &  (1,1,0),(\tau,\tau,0)      &&   (7,7)      &          D\\ &
&                0+,0-,0-     &   2,2&&\\ &
&                0-,0+,1-    &    2,2&&\\
\hline\hline
\mb{{\textbf{Rank 3:}}}&&&&&\\\hline
(3-1)&(0+,0-,0-),(0-,0+,0-)  &  (0,\tau,1),\hphantom{(0,\tau,1),}&&     (12,6)   &     0\\ &
&(\tau,1,0),(1,0,\tau) &&&\\ &
&                0+,0-,0-     &   2&&\\ &
&                0-,0+,0-    &    2&&\\ &
&                0-,0-,0+   &     2&&\\ &
&                0+,\tau-,1-  &      1,1&&\\ &
&                1-,0+,\tau- &       1,1&&\\ &
&                \tau-,1-,0+&        1,1&&\\
\hline

\end{array}$\\
$\begin{array}{|l|lrc|c|c|}
\hline
&G_T&      G_S\hspace*{\fill}&      &         (h^{1,1},h^{2,1})&\pi_1\\ &
 &   \hspace*{\fill}\mb{sectors}  &       \mb{contribution}&         &\\
\hline\hline
(3-2)&(0+,0-,0-),(0-,0+,1-)  & (0,\tau,1),\hphantom{(0,\tau,1),} &&    (12,6)      &  0\\ &
&(\tau,1,0),(1,0,\tau)  &&&\\ &
&                0+,0-,0-       & 1,1&&\\ &
&                0-,0+,1-      &  2&&\\ &
&                0-,\tau-,0+     &   2&&\\ &
&                0+,\tau-,1-    &    2&&\\ &
&                1-,0+,\tau1-  &     1,1&&\\ &
&                \tau-,\tau1-,0+ &      1,1&&\\
\hline

(3-3)&(0+,0-,0-),(0-,0+,0-) &  (1,1,0),  \hphantom{(0,\tau,1),} &&   (17,5) &       0\\ &
&(\tau,\tau,0),(1,\tau,1) &&&\\ &
&                0+,0-,0- &       2&&\\ &
&                0-,0+,0-&        2&&\\ &
&                0-,0-,0+       & 2&&\\ &
&                0+,\tau1-,1-     &  2&&\\ &
&                \tau1-,0+,1-    &   2&&\\ &
&                1-,1-,0+    &    1,1&&\\ &
&                \tau-,\tau-,0+   &     1,1&&\\ &
&                \tau1-,\tau1-,0+&      2&&\\
\hline
(3-4)&(0+,0-,0-),(0-,0+,\tau-) &  (1,1,0), \hphantom{(0,\tau,1),}  &&   (7,7)      &  C\\ &
&(\tau,\tau,0),(1,\tau,1) &&&\\ &
&                0+,0-,0-     &   1,1&&\\ &
&                0-,0+,\tau-    &    1,1&&\\ &
&                0+,\tau1-,1-  &     1,1&&\\ &
&                \tau1-,0+,\tau1- &      1,1&&\\
\hline
(3-5)&(0+,0-,0-),(0-,0+,0-)  & (0,1,1), \hphantom{(0,\tau,1),} &&    (15,3)      &  C\\ &
&(1,0,1),(\tau,\tau,\tau) &&&\\ &
&               0+,0-,0-     &   2&&\\ &
&                0-,0+,0-    &    2&&\\ &
&                0-,0-,0+   &     2&&\\ &
&                0+,1-,1-  &      2&&\\ &
&                1-,0+,1- &       2&&\\ &
&                1-,1-,0+&        2&&\\
\hline
(3-6)&(0+,0-,0-),(0-,0+,\tau-)  & (0,1,1), \hphantom{(0,\tau,1),} &&    (9,9)     &   0\\ &
&(1,0,1),(\tau,\tau,\tau) &&&\\ &
&                0+,0-,0-      &  1,1&&\\ &
&                0-,0+,\tau-     &   1,1&&\\ &
&                \tau-,\tau-,0+    &    1,1&&\\ &
&                0+,1-,1-   &     1,1&&\\ &
&                1-,0+,\tau1- &       1,1&&\\ &
&               \tau1-,\tau1-,0+    &    1,1&&\\
\hline
\end{array}$\\
$\begin{array}{|l|lrc|c|c|}
\hline
&G_T&      G_S\hspace*{\fill}&      &         (h^{1,1},h^{2,1})&\pi_1\\ &
 &   \hspace*{\fill}\mb{sectors}  &       \mb{contribution}&         &\\
\hline\hline
\mb{{\textbf{Rank 4:}}}&&&&&\\\hline
(4-1)&(0+,0-,0-),(0-,0+,0-) &  (0,\tau,1),(\tau,1,0),  && (15,3)  & 0\\ &
 &  (1,0,\tau),(1,1,1)  &&   & \\ &
&                0+,0-,0-    &    1&&\\ &
&                0+,\tau-,1-   &     1&&\\ &
&                0+,1-,\tau1- &      1&&\\ &
&                0+,\tau1-,\tau-&       1&&\\ &
&                0-,0+,0-        &1&&\\ &
&                1-,0+,\tau-       & 1&&\\ &
&                \tau1-,0+,1-     &  1&&\\ &
&                \tau-,0+,\tau1-    &   1&&\\ &
&                0-,0-,0+    &    1&&\\ &
&                \tau-,1-,0+   &     1&&\\ &
&                1-,\tau1-,0+ &      1&&\\ &
&                \tau1-,\tau-,0+&       1&&\\
\hline
\end{array}$\\[1em]
As follows from the discussion in Section \ref{dt}, for 
some of the orbifolds listed above there may be a non-trivial effect
on the resulting Hodge numbers, when twisted actions of $G$ are allowed
on the blow-ups of curves of singularities in $X/G$, that is when discrete
torsion is taken into account. The most popular example of this type
is our model $(0-1)$ which was extensively studied in \cite{vawi95}. 
In that paper, the authors discover that turning on nontrivial 
discrete torsion $\eps\in\Z_2$
in this example of an orbifold by $G\cong(\Z_2)^2$ produces its 
mirror partner. From our discussion it is indeed not hard to check
that the effect of $\eps=-1$ instead of $\eps=1$ is
a swap of the Hodge numbers $h^{1,1},\, h^{2,1}$. 

In general, for each model in our list, an interchange of $h^{1,1}$ and $h^{2,1}$ 
can be achieved by choosing a certain value for discrete torsion. Other types of
discrete torsion exist for some of 
the models, typically producing Hodge number pairs that are intermediate between those of the original orbifold and its mirror. All Hodge pairs obtained this way arise also from other orbifolds without discrete torsion, so they can also be found elsewhere in our table.
we list the possible Hodge numbers below:
$$\begin{array}{|c|c|}
\hline
\mb{model}&\mb{possible values of }(h^{1,1},h^{2,1})\\
\hline\hline
(2-9)& (27,3),\, (15,15),\, (3,27)\\
(3-3)& (17,5),\, (11,11),\, (5,17)\\
(3-5)& (15,3),\, (9,9),\, (3,15)\\
(4-1)& (15,3),\,(12,6),\, (9,9),\, (6,12),\, (3,15)\\
 \hline
\end{array}$$  
Thus we find that the main effect of discrete torsion 
on our  list of possible Hodge numbers $(h^{1,1},h^{2,1})$
is a symmetrization with respect to $h^{1,1}\leftrightarrow h^{2,1}$.
We therefore refrain from a further study of the resulting geometries.
In particular, we do not examine whether there are any Calabi-Yau threefolds 
obtained from allowing non-trivial discrete torsion which
agree with any of the models listed in Table 1, or their mirrors.
\subsubsection*{Table 2: The orbifold family tree}
Below we list all those orbifoldings which are realized as
free quotients between orbifolds listed in Table 1. 
In the diagrams, each entry is of the form
${(r-n)\choose (h^{1,1},h^{2,1})}$, where $(r-n)$ is the label
in Table 1, and $(h^{1,1},h^{2,1})$ gives the corresponding 
Hodge numbers.

We first list all free quotients relating orbifolds $X/G$ with
fundamental group of type $A$:
$$
\begin{array}{ccccccccccc}
 &&{(1-4)\choose (7,7)}&   \longrightarrow &  {(2-8)\choose (5,5)}\\
&\nearrow&&\nearrow\\
{(0,3)\choose (11,11)}&\longrightarrow&{(1-9)\choose (7,7)}\\
&\searrow&&\searrow\\
&&{(1-10)\choose (11,11)}&\longrightarrow&{(2-11)\choose (7,7)}\\[7pt]
\end{array}
$$
Next we list all free quotients relating orbifolds $X/G$ with
fundamental group of type $B$:
$$
\begin{array}{ccccccccccc}
{(0,4)\choose (3,3)} &\longrightarrow&{(1-5)\choose (3,3)}\\
&\searrow\\
&&{(1-11)\choose (3,3)}&\longrightarrow &  {(2-12)\choose (3,3)}\\[7pt]
\end{array}
$$
We list the remaining quotients, where the three columns give orbifolds
with fundamental group $0,\,C,\,D$, respectively:

$$
\begin{array}{ccccccccccc}
{(0-1)\choose (51,3)}  &\longrightarrow &{(1-1)\choose (27,3)} 
&\longrightarrow &{(2-1)\choose (15,3)}\\[7pt]
{(0,2)\choose (19,19)} &\longrightarrow&{(1-3)\choose (11,11)}&   \longrightarrow &  {(2-5)\choose (7,7)}\\
&\searrow&&\nearrow\\
&&{(1-7)\choose (11,11)}&    \longrightarrow & {(2-14)\choose (7,7)}\\[7pt]
{(1-2)\choose (15,15)} & \longrightarrow & {(2-2)\choose (9,9)} \\[7pt]
{(1-6)\choose (31,7)} & \longrightarrow  &{(2-3)\choose (17,5)}\\[7pt]
{(1-8)\choose (15,15)} & \longrightarrow & {(2-7)\choose (9,9)} \\[7pt]
{(2-4)\choose (11,11)}\\[7pt]
{(2-6)\choose (19,7)}\\[7pt]  
{(2-9)\choose (27,3)} & \longrightarrow & {(3-5)\choose (15,3)} \\[7pt]
{(2-10)\choose (11,11)}& \longrightarrow & {(3-4)\choose (7,7)} \\[7pt]
{(2-13)\choose (21,9)}\\[7pt]
{(3-1)\choose (12,6)}\\[7pt]
{(3-2)\choose (12,6)}\\[7pt]
{(3-3)\choose (17,5)}\\[7pt]
{(3-6)\choose (9,9)} \\[7pt]
{(4-1)\choose (15,3)}\\[7pt]
\end{array}
$$
\subsection{On coincidences in the list}\label{coincidences}
The Hodge numbers and fundamental group data in Table 1 do not suffice to
completely distinguish the orbifolds on our list.  We can obtain some
additional topological information:
\begin{lemma} 
The four orbifolds whose fundamental group is an extension of
$(\Z_2)^2$ by $\Z^6$ \mb{\rm(}``type B"\mb{\rm)}, labeled 
\mb{$\rm(0-4), (1-5), (1-11), (2-12)$,}
all with Hodge numbers $(3,3)$, are topologically inequivalent: their
fundamental groups are not isomorphic.
\end{lemma}
\begin{proof} Each of these four orbifolds $X_i$ is a quotient of $\C^3$ by a
group $G_i$ acting without fixed points, so $F_i = \{1\}, G_i = G_i/F_i =
\pi_1(X_i)$.

Each $G_i$ is generated by the two twists $(1+,0-,0-), (0-,1+,1-)$, plus a
rank 6 lattice $L_i$ of translations, with respective generators:
\begin{eqnarray*}
(0-4):&& (2,0,0), (0,2,0), (0,0,2), (2\tau,0,0), (0,2\tau,0), (0,0,2\tau);
\\
(1-5):&& (2,0,0), (0,2,0), (0,0,2), (2\tau,0,0), (0,2\tau,0), (\tau,\tau,\tau);
\\
(1-11):&& (2,0,0), (0,2,0), (0,0,2), (2\tau,0,0), (\tau,\tau,0), (0,0,2\tau);
\\
(2-12):&& (2,0,0), (0,2,0), (0,0,2), (2\tau,0,0), (0,\tau, \tau), (\tau,0,\tau).
\end{eqnarray*}
The commutator $[G,G]$ is generated by:
\begin{eqnarray*}
[(1+,0-,0-), (0-,1+,1-)] &=& (2,-2,2);\\{}
[(1+,0-,0-), L_i] &=& \left\{(0,-2b,-2c) | (a,b,c) \in L_i\right\};\\{}
[(0-,1+,1-), L_i] &=& \left\{(-2a,0,-2c) | (a,b,c) \in L_i\right\}.
\end{eqnarray*}
Consider the four quotients $H_1 = G/[G,G]$ for $G=G_i$. The images of the
two twists square to the non zero elements $(2,0,0)$ and $(0,2,0)$
respectively, which in all four cases are distinct in $H_1$. 
These twists therefore generate a subgroup $(\Z_4)^2$ of $H_1$
which contains the image of the $(\Z_2)^2$ orbit of
the $\Z^3$ of 1-cycles. Therefore, each of the quotients $H_1 = G/[G,G]$ is
a product $H^1 \times H^\tau$ where $H^1 = (\Z_4)^2$ comes from the twists
and the three $1$-cycles, while $H^\tau$ comes from the three $\tau$-cycles.
Explicitly, the four groups $H^\tau$ are: $(\Z_2)^3, \Z_4, (\Z_2)^2,
(\Z_2)^2$, so the four quotients $H_1 = G/[G,G]$ are:
\begin{eqnarray*}
(0-4):&& (\Z_4)^2 \times (\Z_2)^3,\\
(1-5): &&(\Z_4)^3,\\
(1-11): &&(\Z_4)^2 \times (\Z_2)^2,\\
(2-12): &&(\Z_4)^2 \times (\Z_2)^2.
\end{eqnarray*}
It remains to distinguish between the last two orbifolds. For that, note
that in these cases the extension $0 \to L_i \to G_i \to (\Z_2)^2 \to 0$
is uniquely determined by $L_i$. In fact, the projection $G_i \to
(\Z_2)^2$ is just the composition of the abelianization $G \to G/[G,G]$
with multiplication by 2 in the abelian group $G/[G,G]$. So any
isomorphism of the two groups $G_i$ induces an isomorphism of the
extensions, hence of the actions of $(\Z_2)^2$ on $L_i$. This action
sends $(a,b,c) \in L_i$ to $(a,-b,-c), (-a,b,-c),(-a,-b,c)$ respectively.
The sum of the three sublattices fixed under these three involutions is
the lattice $L_{(0-4)}$, which has index 2 in $L_i$ for $i=( 1-11 )$ and
index 4 for $i= ( 2-12 )$, completing the proof that the four
fundamental groups are pairwise non-isomorphic.
\vspace{3pt}
\end{proof}
Unfortunately, comparable topological information is harder to get for
fundamental groups of the other types. For example, one can check that the
extension of $\Z_2$ by $\Z^2$, of ``type A", is unique, with homology
group $H_1(X) = (\Z_2)^3$. This leaves us with the following
undistinguished cases:
$$\begin{array}{|c|c|c|}
\hline
(h^{1,1},h^{2,1})&  \pi_1& \mb{cases}   \\
\hline\hline
(15,15)&                0  &      (1-2),(1-8)\\
(11,11)  &              0    &    (2-4),(2-10)\\
(12,6)  &              0      &  (3-1),(3-2)\\
(11,11) &               A      &  (0-3),(1-10)\\
 (7,7)   &              A       & (1-4),(1-9),(2-11)\\
 (11,11)  &              C    &    (1-3),(1-7)\\
 (9,9)  &              C    &    (2-2),(2-7)\\
 (7,7)  &              D   &    (2-5),(2-14)\\
 \hline
\end{array}$$  
For some of these cases we are able to give a definite answer whether
or not the corresponding threefolds agree. Specifically, 
the models
$(0-3)$ and $(1-10)$ are of same topological type, although their
complex structure is different, as we shall see in Section \ref{BV}.
Together with the orbifold family tree of Table 2
this suggests that the three models $(1-4)$, $(1-9)$, and
$(2-11)$ may be topologically equivalent as well. We know that 
the models $(1-3)$ and $(1-7)$ are distinct as families of complex varieties,
and so are $(2-5)$ and $(2-14)$, as we shall see in Section \ref{schoen}.
It is not clear to us whether they are topologically equivalent.
\section{Free fermion models}\label{ffm}
In this section, we briefly review free fermion models, 
giving the basic structure of those conformal field theories which are obtained
from free fermion constructions. Moreover, we explain how the particular geometric
orbifolds that we have classified in Section \ref{classification} are related to
these models.
\subsection{Model building with free fermions}\label{modbuild}
We use free fermion models to construct  
heterotic string theories in $D=4$ 
dimensions. 
As the name suggests, the basic ingredients to free fermion models are the
representations of the free fermion algebra, see Appendix \ref{ffr}.  
Let $\m H_0,\, \m H_1$ denote the irreducible Fock space representations of the free 
fermion algebra in the NS and the R sector, respectively, enlarged by 
$(-1)^F$ with $F$ the worldsheet fermion number operator.
Roughly, a free fermion model is obtained from an appropriate tensor
product of Fock spaces $\m H_0,\, \m H_1$ by a certain 
projection, whose properties are partly governed by the consistency conditions
of string theory.

In a heterotic theory the left handed side carries at least
$N=1$ supersymmetry, whereas
the right handed side is not supersymmetric. We fermionize
all internal degrees of freedom, thus allowing a description in terms
of free fermions. 
External bosons are not fermionized, since they are free uncompactified
fields, where fermionization does not apply to add degrees of freedom.
The various anomaly cancellation conditions then dictate the following
structure:
In the left handed sector,
we have four external bosons and fermions, two of which are transversal
in the light cone gauge,
$\partial X_\mu,\,\psi^\mu,\,\mu\in\{0,1\}$. Since the superstring critical dimension is $10$, where each coordinate
direction corresponds to three free fermions,
there are
$(10-4)\cdot 3=18$ internal fermions $\chi^i,\,y^i,\,w^i,\,i\in\{1,\ldots,6\}$. The left
handed worldsheet supercurrent, which generates the local conformal transformations, is then given in light cone gauge by: 
$\sum_\mu \nop{\psi^\mu\partial X_\mu}+\sum_i \nop{\chi^i y^i w^i}$ 
\cite{abkw86,gool85,gno85,gko86,dkpr85}. On the right handed
side we have four external bosons, two of which are transversal,
$\qu\partial\qu X_\mu,\,\mu\in\{0,1\}$. Furthermore, given the 
bosonic critical dimension $26$ with each coordinate direction corresponding
to two free fermions,
there are $(26-4)\cdot2=44$
internal fermions $\qu\Phi^i,\,i\in\{1,\ldots,44\}$. All in all, including 
external fermions, we have
$20+44=64$ fermionic degrees of freedom, and we introduce indices
$j\in\{1,\ldots,64\}$ for them in the order
$\psi^0,\,\psi^1,\,\chi^1,\dots,\chi^6,\,y^1,\ldots,y^6,\,
w^1,\ldots,$ $w^6$, $\qu\Phi^1,\ldots,$ $\qu\Phi^{44}$.

To construct a full theory we need to specify which 
combinations of spin structures for each
of the $64$ real free fermions contribute.
Since many of our free fermions result from bosonization, the respective
spin structures are coupled pairwise. This is also
always true for $\psi^0,\,
\psi^1$, to ensure consistency of the coupling with worldsheet gravitinos.
If the $y^j,\,w^j$ combine to six left handed bosons with currents
$i\nop{y^j w^j}$, then each
pair $y^j,\,w^j$ must have coupled spin structures. This is the case for
free fermion models with geometric interpretation on a real six-torus, but
not in general, so we will not assume such couplings in general.
However, if our free fermion
model arises from a heterotic compactification   on a 
Calabi-Yau manifold with
gauge group in $E_8\times E_8$, then among the right handed 
$\qu\Phi^i$ there must be
six pairs yielding the antiholomorphic
partners of the $y^i, w^i$,
which we then denote
$\qu y^i,\,\qu w^i$ instead of $\qu\Phi^1,\ldots,\qu\Phi^{12}$. 
Then $y^i+\qu y^i$ and $w^i+\qu w^i$ are the fields corresponding to the
respective fermionized coordinates, so $y^i, \qu y^i$ and $w^i, \qu w^i$ are
pairs with coupled spin structures. 
The remaining $32$ 
right handed $\qu\Phi^i$ split into two sets $\qu\varphi^i,\,\qu\phi^i$
of eight complex Dirac fermions each, to allow bosonization to $8$ 
independent currents for each $E_8$ summand of the gauge Kac-Moody algebra, 
cf.\ \cite{abkw86,gool85,gno85,gko86}. 
Each Dirac fermion is equivalent to
one boson, so the notation implies that the real and imaginary parts
of $\qu\varphi^i,\,\qu\phi^i$  also
have coupled spin structures. Finally, the 
heterotic theory on a Calabi-Yau manifold
has $(N,\qu N)=(2,0)$ worldsheet supersymmetry with $U(1)$ current given
by, say, $\sum i\nop{\chi^{2j-1}\chi^{2j}}$, which is
reflected in a pairwise coupling of spin structures for the $\chi^i$, here 
$(\chi^1,\chi^2),\,(\chi^3,\chi^4),\,(\chi^5,\chi^6)$. 
%
%

For what follows we therefore assume that
the $64$ fermions are arranged 
into pairs with coupled spin structures. In other words, we fix an injective vectorspace
homomorphism
$\iota:\F_2^{32}\longrightarrow\F_2^{64}$ and
a fixpoint free involution
$\sigma$ on $\{1,\ldots,64\}$ such that for all $\alpha\in\im(\iota)$
and all $i\in\{1,\ldots,64\}$ we have $\alpha_i=\alpha_{\sigma(i)}$.
For $\alpha\in\im(\iota)$ we then define
$$
\m H_\alpha:=\mb{pr}\left(\bigotimes_{i=1}^{64} \m H_{\alpha_i}^i\right),
$$
where\footnote{Here and in the following, indices $j$
serve as a reminder that we are considering the $j^{\mbox{\tiny{}th}}$
free fermion,  $j\in\{1,\ldots,64\}$.} 
$\mb{pr}$ acts trivially on all $\m H_0^i$ but projects 
$\m H_1^j\otimes\m H_1^{\sigma(j)}$ onto a chosen irreducible representation
of the algebra 
generated by the $\psi_a^j,\,\psi_a^{\sigma(j)}$ with $a\in\Z$ and the 
total worldsheet fermion number operator $(-1)^{F_j+F_{\sigma(j)}}$
(see Appendix \ref{ffr}).

While a single fermion can live in one of two sectors (NS or R), states in
our full theory fall into sectors ${\m H}_\alpha$ characterized by
$\alpha\in\im(\iota)\subset\F_2^{64}$. 
By the above, the relevant contributions to the
partition function have the form
\beqn{trace}
Z\!\!\left[\!\!\begin{array}{c}\alpha\\\beta\end{array}\!\!\right]\!\!(\tau,\qu\tau) 
&:=& \prod_{j=1}^{20} 
Z\!\!\left[\!\!\begin{array}{c}\alpha_j\\\beta_j\end{array}\!\!\right]\!\!(\tau)
\prod_{j=21}^{64} Z\!\!\left[\!\!\begin{array}{c}\alpha_j\\\beta_j\end{array}\!\!\right]\!\!(\qu\tau)\nonumber\\
&\stackrel{\req{traceformula}}{=}&
\tr[{\m H}_\alpha]\left[\prod_j(-1)^{\beta_j F_j}
q^{L_0-c/24}\, \qu q^{\qu L_0-\qu c/24}\right],
\eeqn
where $L_0=\sum_{j=1}^{20}L_0^j,\,\qu L_0=\sum_{j=21}^{64}\qu L_0^j$
are the total left and right handed Virasoro zero modes and 
$c=10,\,\qu c=22$ the total central charges. In the following we 
also abbreviate
$$
\prod_j(-1)^{\beta_j F_j} = e^{\pi i\beta\cdot F}.
$$
As mentioned above, 
our models include two left-handed external fermions, 
$\psi^\mu,\,\mu\in\{0,1\}$, and consistency of their coupling to
worldsheet gravitinos translates into
the requirement that  these two 
must always have the same spin structures. 
This means that we can assign a definite spin statistic 
$\delta_\alpha\in\{\pm1\}$ to each sector
${\m H}_\alpha$ which contributes to
our theory:
$$
\fa\alpha\in\F_2^{64}\mbox{ with }\alpha_1=\alpha_2\colon\quad
\delta_\alpha:= (-1)^{\alpha_1},
$$
where $\alpha\in\im(\iota)$ implies $\alpha_1=\alpha_2$.
By definition,
sectors ${\m H}_\alpha$ with positive spin statistic $\delta_\alpha=1$
contain the spacetime bosons, whereas
sectors ${\m H}_\alpha$ with negative spin statistic $\delta_\alpha=-1$
contain the spacetime fermions.
We can thus use the spacetime fermion number operator $F_S$, where
$(-1)^{F_S}$ acts trivially on spacetime bosons and by multiplication with
$(-1)$ on spacetime fermions. We can now make an ansatz for the 
Hilbert space $\m H$ and the partition 
function 
\beq{genpf}
Z(\tau,\qu\tau) =
\tr[{\m H}] \left[(-1)^{F_S}
q^{L_0-c/24}\, \qu q^{\qu L_0-\qu c/24}\right]
\eeq
of the total fermionized theory. Namely, we begin by requiring that $Z$ has the form 
\beq{pf}
Z(\tau,\qu\tau) 
= {\textstyle{1\over|{\m F}|}}\sum_{\alpha,\beta\in{\m F}}  
C\!\!\left[\!\!\begin{array}{c}\alpha\\\beta\end{array}\!\!\right]Z\!\!\left[\!\!\begin{array}{c}\alpha\\\beta\end{array}\!\!\right]\!\!(\tau,\qu\tau), 
\quad
C\!\!\left[\!\!\begin{array}{c}\alpha\\\beta\end{array}\!\!\right]\in\Z,
\eeq
for some ${\m F}\subset\F_2^{64}$ such that for every 
$\alpha\in{\m F}$ there is a $\beta\in{\m F}$ with 
${C\!\!\left[\!\!\begin{array}{c}\alpha\\\beta\end{array}\!\!\right]\neq0}$.
Note that if $Z$ is obtained from some pre-Hilbert space 
$\m H$ by \req{genpf}, then $Z$
vanishes iff in $\m H$ there is a 
$1\colon1$ correspondence between spacetime bosons and spacetime fermions
which respects the $(L_0,\qu L_0)$ eigenvalues. 
In other words, 
$Z(\tau,\qu\tau)\equiv0$ iff our theory possesses spacetime supersymmetry, see also Section \ref{SUSY}.

Let us now give a description of $\m H$ which yields \req{genpf} with \req{pf}, and
let us discuss appropriate
restrictions on the coefficients
$C\!\!\left[\!\!\begin{array}{c}\alpha\\\beta\end{array}\!\!\right]$.
First rewrite \req{pf} with \req{trace} to find
\beqn{sector}
Z(\tau,\qu\tau)
=\sum_{\alpha\in{\m F}} \delta_\alpha Z_\alpha(\tau,\qu\tau), \quad
Z_\alpha(\tau,\qu\tau)
&=&
\tr[{\m H}_\alpha]\left[P^{\m F}
q^{L_0-c/24}\, \qu q^{\qu L_0-\qu c/24}\right],\nonumber\\
P^{\m F}\colon \bigoplus_{\alpha\in\m F}{\m H}_\alpha \longrightarrow 
\bigoplus_{\alpha\in\m F} {\m  H}_\alpha,
\quad\quad\;\; P^{\m F}_{\mid \m H_\alpha} 
&:=& {\textstyle{1\over|{\m F}|}}\sum_{\beta\in{\m F}} 
\delta_\alpha C\!\!\left[\!\!\begin{array}{c}\alpha\\\beta\end{array}\!\!\right]
\cdot e^{\pi i\beta\cdot F},
\eeqn
where by construction $(-1)^{F_S}_{\mid\m H_\alpha}=\delta_\alpha$.
Then \req{genpf} holds with ${\m H}= P^{\m F}\oplus_{\alpha\in{\m F}} {\m H}_\alpha$.

We wish to interpret $P^{\m F}$ 
as a projection operator. 
To ensure that $P^{\m F}\circ P^{\m F}=P^{\m F}$ one checks that it 
suffices to assume that $\m F\subset\F_2^{64}$ is a vector space and
that 
$$
\fa\alpha,\beta,\gamma\in\m F\colon\quad\quad
C\!\!\left[\!\!\begin{array}{c}\alpha\\\beta+\gamma\end{array}\!\!\right]
= \delta_\alpha\, C\!\!\left[\!\!\begin{array}{c}\alpha\\\beta\end{array}\!\!\right]
C\!\!\left[\!\!\begin{array}{c}\alpha\\\gamma\end{array}\!\!\right].
$$
In other words,
$$
\fa\alpha\in\m F\colon\quad\quad
\chi_\alpha:\m F\longrightarrow\{\pm1\}, \quad
\chi_\alpha(\beta):=
\delta_\alpha\delta_\beta
C\!\!\left[\!\!\begin{array}{c}\alpha\\\beta\end{array}\!\!\right]
$$
is a character. For later convenience we introduce the notation
$$
\chi:\m F\times\m F\longrightarrow\{\pm1\}, \quad
\chi\!\!\left[\!\!\begin{array}{c}\alpha\\\beta\end{array}\!\!\right]
:=
\delta_\alpha\delta_\beta
C\!\!\left[\!\!\begin{array}{c}\alpha\\\beta\end{array}\!\!\right]
$$
and indeed assume for the following that $\m F$ is a vector space and
that
\beq{char}
\fa\alpha,\beta,\gamma\in\m F\colon\quad\quad
\chi\!\!\left[\!\!\begin{array}{c}\alpha\\\beta+\gamma\end{array}\!\!\right]
= \chi\!\!\left[\!\!\begin{array}{c}\alpha\\\beta\end{array}\!\!\right]
\chi\!\!\left[\!\!\begin{array}{c}\alpha\\\gamma\end{array}\!\!\right]\in
\{\pm1\}.
\eeq
Then the remaining restrictions
on the possible choices of 
$\chi\!\!\left[\!\!\begin{array}{c}\alpha\\\beta\end{array}\!\!\right]$ 
will come from the fact that the partition function $Z$
must be modular invariant and that all fields in 
${\m H}=P^{\m  F}{\bigoplus\limits_{\alpha\in\m F}}{\m H}_\alpha$ 
must be pairwise semi-local.
We introduce
\beqn{skp}
\fa\alpha,\beta\in\im(\iota)\colon\;\;
\alpha\cdot\beta:={1\over2}\sum_{j=1}^{20} \alpha_j\beta_j
- {1\over2}\sum_{j=21}^{64} \alpha_j\beta_j
&=& \alpha_L\cdot\beta_L - \alpha_R\cdot\beta_R,\;\\
\alpha^2&:=&\alpha\cdot\alpha\nonumber
\eeqn
with $\alpha_L,\,\beta_L\in\F_2^{20},\;\alpha_R,\,\beta_R\in\F_2^{44}$.
This induces a scalar product on $\F_2^{32}\cong\im(\iota)$ 
with signature $(n_+,n_-,n_0)$. The form \req{skp}
encodes the conformal dimensions of the
ground states in ${\m H}_\alpha$ for $\alpha\in\m F$ if we 
lift $\m F\subset\F_2^{64}$ to $\Z^{64}$ with entries in 
$\{0,1\}$. Namely, 
since NS ground states 
in our Fock space representation of the free fermion algebra have
vanishing conformal dimension, whereas the R ground states of a single
free fermion have dimension ${1\over16}$, we see that the
ground states of ${\m H}_\alpha$ have conformal dimensions
\beq{gsen}
(h,\qu h) = 
\left( {\alpha_L\cdot\alpha_L\over8}, {\alpha_R\cdot\alpha_R\over8}
\right).
\eeq
To get a well-defined fermionic theory, namely to ensure semi-locality,
all conformal spins have to 
be half integer, $h-\qu h\in{1\over2}\Z$. Since on ${\m H}_\alpha$,
the condition $h-\qu h\in{1\over2}\Z$ depends solely on the conformal
spin of the ground state as obtained from \req{gsen}, we thus need
\beq{levelmatching}
\fa\alpha\in{\m F}\colon\quad
\alpha^2\equiv0(4).
\eeq
Vice versa, if this constraint holds, then one checks that on $\m H$
a so-called OPE can be introduced consistently, as is necessary to
construct a CFT.

Every theory must contain a vacuum sector ${\m H}_0$, as follows from 
$0\in{\m F}$, and uniqueness of the vacuum dictates 
$C\!\!\left[\!\!\begin{array}{c}0\\0\end{array}\!\!\right]=1$. 
By the transformation properties listed 
in Appendix \ref{thetafctns}, the modular transformation
$\tau\mapsto-{1\over\tau+1}$ maps 
$Z\!\!\left[\!\!\begin{array}{c}0\\0\end{array}\!\!\right]$ to
$Z\!\!\left[\!\!\begin{array}{c}1\\0\end{array}\!\!\right]$. Since
$Z$ must be modular invariant, 
this implies that $0\in{\m F}$ entails $1\in{\m F}$,
the vector with all $64$ entries given by $1$. Hence 
by \req{levelmatching} we need 
$n_+-n_-\equiv0(4)$. In fact, assuming \req{char}, modular
invariance of the partition function holds if we also impose
\beq{trafo}
C\!\!\left[\!\!\begin{array}{c}\alpha\\\beta\end{array}\!\!\right]
=C\!\!\left[\!\!\begin{array}{c}\beta\\\alpha\end{array}\!\!\right],
\quad\quad
C\!\!\left[\!\!\begin{array}{c}\alpha\\\alpha+1\end{array}\!\!\right]
= -\delta_\alpha\,e^{-\pi i\alpha^2/4} .
\eeq
The latter condition is only consistent if $n_+-n_-\equiv4(8)$. 
These rules now allow us to restrict
attention to $C\!\!\left[\!\!\begin{array}{c}b_i\\b_j\end{array}\!\!\right]$
with $b_i,\,b_j\in B$ a basis of ${\m F}$.
For example,
\beq{examples}
\begin{array}{rclcrcl}
B_0&=&\{1\},\\
B_{SUSY}&=&\{1,s\}\quad
&\mbox{ with }
&s_i&=&\left\{\begin{array}{rl}1&\mbox{ if } i\leq8,\\
0&\mbox{ otherwise,}\end{array}\right.\\[1em]
B_{tor}&=&\{1,s,\xi_1,\xi_2\}\quad
&\mbox{ with }
&(\xi_1)_i&=&\left\{\begin{array}{rl}1&\mbox{ if } i>48,\\
0&\mbox{ otherwise,}\end{array}\right.\\[1em]
&&&&(\xi_2)_i&=&\left\{\begin{array}{rl}1&\mbox{ if } 32<i\leq48,\\
0&\mbox{ otherwise.}\end{array}\right.
\end{array}
\eeq
Consistent choices for the values of $C$ are obtained, e.g., as follows:
$$
\begin{array}[b]{|l||c|c|c|c|c|}
\hline\vphantom{\sum_{i\over2}^{k\over2}}
\;\;\beta\!\!\!\!&0&1&s&\xi_1&\xi_2\\[-0.2em]
\alpha&&&&&\\
\hline\hline\vphantom{\sum_{i\over2}^{k\over2}}
0&1&-1&-1&1&1\\\hline\vphantom{\sum_{i\over2}^{k\over2}}
1&-1&-1&-1&(-1)^{\eps_1}&(-1)^{\eps_2}\\\hline\vphantom{\sum_{i\over2}^{k\over2}}
s&-1&-1&-1&(-1)^{\eps_1^\prime}&(-1)^{\eps_2^\prime}\\\hline\vphantom{\sum_{i\over2}^{k\over2}}
\xi_1&1&(-1)^{\eps_1}&(-1)^{\eps_1^\prime}&(-1)^{\eps_1}&(-1)^{\kappa}\\
\hline\vphantom{\sum_{i\over2}^{k\over2}}
\xi_2&1&(-1)^{\eps_2}&(-1)^{\eps_2^\prime}&(-1)^{\kappa}&(-1)^{\eps_2}\\
\hline
\end{array}\quad\quad
\eps_i,\,\eps_i^\prime,\kappa\in\{0,1\}.
$$
\subsection{GSO and orbifold projections}\label{gsoorb}
In general, from the partition function \req{genpf} of a CFT we can read
the net number of (spacetime)
bosonic minus fermionic states in ${\m H}$
which are eigenvectors of the Virasoro zero modes $L_0,\,\qu L_0$ with any
pair of eigenvalues $h,\,\qu h\in\R$. 
The special structure of the partition
function \req{sector} allows us to determine the numbers of bosonic and 
fermionic contributions separately by hand. Namely,
for each $\beta\in{\m F}$, \req{trafo} implies
\beqn{proj}
\fa\alpha,\gamma\in{\m F}\colon\quad
C\!\!\left[\!\!\begin{array}{c}\alpha\\\gamma\end{array}\!\!\right]
Z\!\!\left[\!\!\begin{array}{c}\alpha\\\gamma\end{array}\!\!\right]
&+&
C\!\!\left[\!\!\begin{array}{c}\alpha\\\gamma+\beta\end{array}\!\!\right]
Z\!\!\left[\!\!\begin{array}{c}\alpha\\\gamma+\beta\end{array}\!\!\right]
\nonumber\\
&=&
C\!\!\left[\!\!\begin{array}{c}\alpha\\\gamma\end{array}\!\!\right]
\left(Z\!\!\left[\!\!\begin{array}{c}\alpha\\\gamma\end{array}\!\!\right]
\;+\; \delta_\alpha\,
C\!\!\left[\!\!\begin{array}{c}\alpha\\\beta\end{array}\!\!\right]
Z\!\!\left[\!\!\begin{array}{c}\alpha\\\gamma+\beta\end{array}\!\!\right]
\right).
\eeqn
Since $Z\!\!\left[\!\!\begin{array}{c}\alpha\\\gamma\end{array}\!\!\right]$
and
$Z\!\!\left[\!\!\begin{array}{c}\alpha\\\gamma+\beta\end{array}\!\!\right]$
differ by an insertion of $(-1)^{F_j}$ in the trace $\tr[{\m H}_\alpha]$
in each component $j$ where $\beta_j=1$, 
this amounts to projecting onto states 
$|\sigma\rangle_\alpha\in
{\m H}_\alpha$ which obey
\beq{gso}
\delta_\alpha\,
C\!\!\left[\!\!\begin{array}{c}\alpha\\\beta\end{array}\!\!\right]
e^{\pi i \beta\cdot F} |\sigma\rangle_\alpha = |\sigma\rangle_\alpha.
\eeq
The condition \req{gso} is often called GSO-projection. Since 
\req{trafo} also implies that 
$$
\fa\alpha\in{\m F}\colon\quad
C\!\!\left[\!\!\begin{array}{c}\alpha\\0\end{array}\!\!\right]
=\delta_\alpha,
$$
so that \req{gso} trivially holds for $\beta=0$, and since
${\m F}$
is a vector space over $\F_2$, we see that $|\sigma\rangle_\alpha$
remains in the spectrum of our theory iff \req{gso} holds for all
basis elements $\beta=b_j\in B$.

The interpretation of $P^{\m F}$ as projection operator allows us to
relate different choices of bases $B,\, B^\prime$ by orbifolding. 
Namely, if $B\subset B^\prime$ with $r=|B|$ and
$r^\prime=|B^\prime|$, then the corresponding theories are related by
orbifolding with respect to a group $G$ of type $(\Z_2)^{r^\prime-r}$,
as long as the coefficients 
$C\!\!\left[\!\!\begin{array}{c}\alpha\\\beta\end{array}\!\!\right]$ for
$\alpha,\beta\in\span_{\F_2}B$ agree. 
To see this, let
$$
B=\{b_1,\ldots,b_r\} \mb{ and }
B^\prime=\{b_1,\ldots,b_r,b_{r+1},\ldots,b_{r^\prime}\}
=B\cup B^\perp.
$$
Define
$$
{\m F}:= \span_{\F_2}B,\quad
{\m F}^\prime:= \span_{\F_2}B^\prime,\quad
{\m F}^\perp:= \span_{\F_2}B^\perp.
$$
Recall that the CFTs ${\m C}, {\m C}^\prime$ corresponding to the bases
$B$ and $B^\prime$, respectively, have underlying pre-Hilbert spaces
$$
{\m H} = P^{\m F}\bigoplus_{\alpha\in{\m F}} {\m H}_\alpha,\quad
{\m H}^\prime = P^{\m F^\prime}\bigoplus_{\alpha^\prime\in{\m F^\prime}} {\m H}_{\alpha^\prime}.
$$
We now wish to reinterpret ${\m H}^\prime$ as arising from $\m H$ by orbifolding,
i.e.\ by rewriting 
$$
{\m H}^\prime = {\m H}^G\oplus \left( {\m H}^{twist} \right)^G,
$$
where a superscript $G$ denotes the $G$-invariant subspace of a given
vectorspace, and 
${\m H}^{twist} := \bigoplus _{\gamma \neq 0} {\m H}_\gamma$ 
is the sum of the twisted sectors. We will see that
this amounts to a simple reordering of the summands of ${\m H}^\prime$. 
Indeed,
the sectors $\m H_{\alpha^\prime},\, \alpha^\prime\in\m F^\prime$, which contribute to
the theory $\m C^\prime$ associated to $B^\prime$ can be listed as
follows:
$$
\fa \gamma\in\m F^\perp:\quad
\wt{\m H}_\gamma^{orb} := \bigoplus_{\wt\alpha\in\m F} \m H_{\wt\alpha+\gamma},
\quad\mb{so }
{\m H}^\prime = P^{\m F^\prime}\bigoplus_{\gamma\in\m F^\perp} \wt{\m H}_{\gamma}^{orb}.
$$
Using $P^{\m F^\prime}=P^{\m F^\perp}\circ P^{\m F}$ and 
for all $\gamma\in\m F^\perp$:
${\m H}_\gamma^{orb} := P^{\m F}\wt{\m H}_{\gamma}^{orb}$, we observe
${\m H}^\prime = P^{\m F^\perp}\bigoplus_{\gamma\in\m F^\perp} {\m H}_{\gamma}^{orb}$.
It thus remains to argue that $P^{\m F^\perp}{\m H}_{\gamma}^{orb}={\m H}_\gamma^G$
for $\gamma\in{\m F}^\perp$ with ${\m F}^\perp$ acting as orbifolding group
$G\cong(\Z_2)^{r^\prime-r}$ and ${\m H}_\gamma$ the $\gamma$-twisted sector. 
First, one immediately checks ${\m H}^{orb}_0={\m H}$,
the original pre-Hilbert space of the theory $\m C$ associated to
$B$, as necessary. On each 
$\m H_\gamma^{orb}$ we have an action of a group
$G\cong(\Z_2)^{r^\prime-r}$ which is generated by $g_{b_j}$ with $b_j\in
B^\perp$, where 
\beq{orb}
\mb{for } |\sigma\rangle_\alpha\in\m H_\alpha\subset\m H_\gamma^{orb}:
\quad\quad
g_{b_j} |\sigma\rangle_\alpha
:= \delta_{\alpha}\,
C\!\!\left[\!\!\begin{array}{c}\alpha\\b_j\end{array}\!\!\right]
e^{\pi i b_j\cdot F} |\sigma\rangle_\alpha .
\eeq
Invariance under $G$ then is
equivalent to $|\sigma\rangle$ obeying \req{gso} for all $\beta\in\m
F^\perp$. 
Hence $P^{\m F^\perp}$ is the projection onto $G$-invariant states, as
claimed. One also checks that ${\m H}_\gamma^{orb}$ is indeed a 
$\gamma$-twisted representation of the OPE of $\m H$.

To obtain a well-defined $G$ orbifold CFT of $\m C$
we cannot allow arbitrary $G$ actions as in \req{orb}. Namely, 
$G$ must obey the
so-called level matching conditions \cite{va86}. These conditions have
been translated into the language of free fermion constructions in
\cite{muwi86}. Note that for even $n$, in \cite[(7)]{muwi86} the first
condition gives a necessary and sufficient condition for level
matching \cite[(5)]{muwi86} to hold. Since there, external fermions
$\psi^\mu,\, \mu\in\{0,1\}$, are never twisted, in our notation these
conditions read $\alpha^2\equiv0(4)$ for all
$\alpha\in\m F^\perp$ and hence are equivalent to our condition
\req{levelmatching} on $\m F^\perp$. The remaining conditions
\cite[(7)]{muwi86} are necessary to ensure that $G$ in fact acts as
$(\Z_2)^{r^\prime-r}$ on $\m H_0^{orb}$. In our language they
guarantee that also $2\alpha\cdot\beta\equiv0(4)$ for all 
$\alpha\in\m F,\, \beta\in\m F^\perp$, i.e.\ that all 
$\gamma\in\m F\oplus\m F^\perp$ obey $\gamma^2\equiv0(4)$.

It should be kept in mind that the orbifoldings for free fermion models 
$\m C$ described above in any given interpretation of  $\m C$
as a nonlinear sigma model on some Calabi-Yau variety
do not necessarily translate into geometric orbifoldings of that variety.
\subsection{Supersymmetry}\label{SUSY}
Consider a free fermion model specified by a choice of $\m F$ and of
coefficients $C$. Assume that for every $\beta\in{\m F}$
we have $\beta_1=\cdots=\beta_8$. We claim that the  
theory is automatically spacetime supersymmetric, i.e.
$Z(\tau,\qu\tau)\equiv0$, if $s\in{\m F}$, and if
$$
\chi\!\!\left[\!\!\begin{array}{c}s\\\beta\end{array}\!\!\right]
=\delta_\beta,\quad\mb{i.e.}\quad
C\!\!\left[\!\!\begin{array}{c}s\\\beta\end{array}\!\!\right]
= -1 \quad\mbox{ for all }\quad \beta\in{\m F},
$$
amounting to $\eps^\prime_i=1$ in the table below \req{examples}.
In terms of the partition function this can be seen as
follows: By assumption,
${\m F}={\m F}^0\cup{\m F}^1$ where $\delta_\beta=(-1)^b$ for
$\beta\in{\m F}^b$, and $\beta\in{\m F}^0\Leftrightarrow
\beta+s\in{\m F}^1$. Moreover, 
$$
\fa\alpha\in{\m F}^a,\,\beta\in{\m F}^b\colon\quad
Z\!\!\left[\!\!\begin{array}{c}\alpha\\\beta\end{array}\!\!\right]\!\!(\tau,\qu\tau)
= Z\!\!\left[\!\!\begin{array}{c}a\\b\end{array}\!\!\right]^8\!\!\!(\tau)\cdot
\prod_{j=9}^{20} 
Z\!\!\left[\!\!\begin{array}{c}\alpha_j\\\beta_j\end{array}\!\!\right]\!\!(\tau)
\prod_{j=21}^{64} Z\!\!\left[\!\!\begin{array}{c}\alpha_j\\\beta_j\end{array}\!\!\right]\!\!(\qu\tau).
$$
Hence by \req{trafo}
\beqn{susy}
Z(\tau,\qu\tau)&=&{\textstyle{1\over|{\m F}|}} \sum_{\alpha,\beta\in{\m F}^0}
\left\{ 
C\!\!\left[\!\!\begin{array}{c}\alpha\\\beta\end{array}\!\!\right]
Z\!\!\left[\!\!\begin{array}{c}\alpha\\\beta\end{array}\!\!\right]\!\!(\tau,\qu\tau)
+
C\!\!\left[\!\!\begin{array}{c}\alpha\\\beta+s\end{array}\!\!\right]
Z\!\!\left[\!\!\begin{array}{c}\alpha\\\beta+s\end{array}\!\!\right]\!\!(\tau,\qu\tau)\right.\nonumber\\
&&\hphantom{{1\over|{\m F}|} \sum_{\alpha,\beta\in{\m F}^0}}
\left.+\,
C\!\!\left[\!\!\begin{array}{c}\alpha+s\\\beta\end{array}\!\!\right]
Z\!\!\left[\!\!\begin{array}{c}\alpha+s\\\beta\end{array}\!\!\right]\!\!(\tau,\qu\tau)
+
C\!\!\left[\!\!\begin{array}{c}\alpha+s\\\beta+s\end{array}\!\!\right]
Z\!\!\left[\!\!\begin{array}{c}\alpha+s\\\beta+s\end{array}\!\!\right]\!\!(\tau,\qu\tau)
\right\} \nonumber\\
&=&{\textstyle{1\over|{\m F}|}} \sum_{\alpha,\beta\in{\m F}^0}
C\!\!\left[\!\!\begin{array}{c}\alpha\\\beta\end{array}\!\!\right]
\left\{ 
Z\!\!\left[\!\!\begin{array}{c}\alpha\\\beta\end{array}\!\!\right]\!\!(\tau,\qu\tau)
+
\delta_\alpha
C\!\!\left[\!\!\begin{array}{c}\alpha\\s\end{array}\!\!\right]
Z\!\!\left[\!\!\begin{array}{c}\alpha\\\beta+s\end{array}\!\!\right]\!\!(\tau,\qu\tau)\right.\nonumber\\
&&\hphantom{{1\over|{\m F}|} \sum_{\alpha,\beta\in{\m F}^0}}
\left.+\,
\delta_\beta
C\!\!\left[\!\!\begin{array}{c}s\\\beta\end{array}\!\!\right]
Z\!\!\left[\!\!\begin{array}{c}\alpha+s\\\beta\end{array}\!\!\right]\!\!(\tau,\qu\tau)
+\delta_\beta
C\!\!\left[\!\!\begin{array}{c}s\\\beta\end{array}\!\!\right]
Z\!\!\left[\!\!\begin{array}{c}\alpha+s\\\beta+s\end{array}\!\!\right]\!\!(\tau,\qu\tau)
\right\} \nonumber\\
&=&{\textstyle{1\over|{\m F}|}} \sum_{\alpha,\beta\in{\m F}^0}
C\!\!\left[\!\!\begin{array}{c}\alpha\\\beta\end{array}\!\!\right]
\left\{ 
Z\!\!\left[\!\!\begin{array}{c}0\\0\end{array}\!\!\right]^8\!\!\!(\tau)
-Z\!\!\left[\!\!\begin{array}{c}0\\1\end{array}\!\!\right]^8\!\!\!(\tau)
-Z\!\!\left[\!\!\begin{array}{c}1\\0\end{array}\!\!\right]^8\!\!\!(\tau)
-Z\!\!\left[\!\!\begin{array}{c}1\\1\end{array}\!\!\right]^8\!\!\!(\tau)
\right\} \nonumber\\
&&\hphantom{{1\over|{\m F}|} \sum_{\alpha,\beta\in{\m F}^0}}
\cdot\prod_{j=9}^{20} 
Z\!\!\left[\!\!\begin{array}{c}\alpha_j\\\beta_j\end{array}\!\!\right]\!\!(\tau)
\prod_{j=21}^{64} Z\!\!\left[\!\!\begin{array}{c}\alpha_j\\\beta_j\end{array}\!\!\right]\!\!(\qu\tau)\\
&\stackrel{\req{vanish}}{=}&
0.\nonumber
\eeqn
%
%
%
\subsection{Gauge algebras}\label{gagr}
The gauge algebra of a free fermion model is the Lie algebra
generated by  the zero modes of its gauge bosons. 
The gauge bosons are those massless spacetime
bosons $\Phi$ which generate deformations 
$\delta_\Phi S = \Phi(z,\qu z) dz\wedge d\qu z$ of the action of 
our free fermion
CFTs and which transform appropriately under the action of the ``external''
space-time Lorentz group. In particular this means that we are 
looking for fields $\Phi$ with conformal weights $h=\qu h=1$ which 
transform in the vector representation of the Lorentz group. 
In the literature, fields with $h=\qu h=1$ are often called massless,
since $\delta_\Phi S$ is invariant under infinitesimal conformal 
transformations of $\C$, i.e. $\delta_\Phi S$ 
defines a ``massless deformation'' of the theory. To
preserve the left handed worldsheet supersymmetry of our heterotic
CFT, the field $\Phi$ must also be the top entry of an 
$(N,\qu N)=(1,0)$ supermultiplet. So equivalently
to listing the appropriate massless fields  we can count states with 
$(h,\qu h)=({1\over2},1)$ in our CFT, the lowest components of multiplets
containing a massless field, if we keep in mind that we need to 
apply a left handed worldsheet supersymmetry to obtain the actual 
massless state.

Again, the particular form of our models allows us to count these
states by hand. Namely, by the discussion of Section \ref{modbuild},
each state in our theory belongs to a sector ${\m H}_\alpha$ with
$\alpha\in\im(\iota)\subset\F_2^{64}$, 
and $\alpha_j=1$ iff the $j^{\mbox{\tiny{}th}}$
free fermion belongs to the R-sector. Using \req{gsen} to determine 
the conformal dimensions of the ground states in ${\m H}_\alpha$
we see:
To list all gauge bosons we need to list all states in 
our theory which are obtained from the ground state of ${\m H}_\alpha$
by the action of creation operators, which obey the GSO projection
\req{gso}, and such that
\beq{massless}
{\alpha_L\cdot\alpha_L\over8}+N_L={1\over2},\quad 
{\alpha_R\cdot\alpha_R\over8}+N_R=1.
\eeq
Here $N_L,\,N_R\in{1\over2}\N$ count the energy coming from creation
operators with integer (half integer) contributions from the
$j^{\mbox{\tiny{}th}}$ component iff $\alpha_j=1$ ($\alpha_j=0$), i.e.
iff the $j^{\mbox{\tiny{}th}}$ free fermion is in the R (NS) sector.
$N_R$ can also get contributions from the two right moving 
bosons $\qu\partial\qu X_\mu,\,\mu\in\{0,1\}$, whose creation operators
are integer moded. 
Finally, we have to inspect the contributions coming from the external
fields to find all massless fields that transform in the 
vector representation of the external Lorentz group. 

Let us count the gauge bosons in two of the examples listed in
\req{examples} with the choices of $C$ given there:

For $B_0=\{1\}$, we have ${\m F}=\{0,1\}$, and ${\m H}_1$ does not
contain massless states since the ground state already has 
conformal dimensions $(h,\qu h) = ({10\over8},{22\over8})$, 
i.e. $h>{1\over2}$. In fact, ${\m H}_1$ is spacetime fermionic and as
such cannot contain gauge bosons anyway.
For ${\m H}_0$, the GSO condition \req{gso}
enforces
$$
-e^{\pi i 1\cdot F}|\sigma\rangle_0 = |\sigma\rangle_0,
$$
i.e. the total number of fermionic creation operators must be odd.
Together with $N_L={1\over2},\,N_R=1$ from \req{massless}, since
$\alpha=0$ leaves all free fermions in the NS sector, we must have 
one creation operator on the left and two fermionic or one bosonic 
one on the right handed side. This yields the following
fields:
\beqno
&&\psi^\mu\qu\partial\qu X_\nu,\;
\chi^i\qu\partial\qu X_\mu,
y^i\qu\partial\qu X_\mu,
w^i\qu\partial\qu X_\mu\;\; (i\in\{1,\ldots,6\}),\\
&&\psi^\mu\qu\Phi^i\qu\Phi^j,\;
\chi^i\qu\Phi^j\qu\Phi^k,y^i\qu\Phi^j\qu\Phi^k,w^i\qu\Phi^j\qu\Phi^k\;\;
 (i\in\{1,\ldots,6\},\,j,k\in\{1,\ldots,44\}).
\eeqno
Counting only those fields which transform in the vector
representation of the external Lorentz group, we find the following
gauge bosons:
\beqno
\chi^i\qu\partial\qu X_\mu,
y^i\qu\partial\qu X_\mu,
w^i\qu\partial\qu X_\mu 
&\mbox{giving}&
\fs\fo(3)^6=\fs\fu(2)^6,\\
\psi^\mu\qu\Phi^i\qu\Phi^j
&\mbox{giving}&
\fs\fo(44).
\eeqno
Moreover, $\psi^\mu\qu\partial\qu X_\nu$
also gives a massless field (the graviton), and so do
$\chi^i\qu\Phi^j\qu\Phi^k$, $y^i\qu\Phi^j\qu\Phi^k$, 
$w^i\qu\Phi^j\qu\Phi^k$
(Lorentz scalars in the $(3)_i\times\mbox{ad}_{\fs\fo(44)}$).

For $B_{tor}$ and with $\eps_i^\prime=1$ the theory is spacetime
supersymmetric, as explained in Section \ref{SUSY}. 
To find gauge bosons, it suffices to consider those
${\m H}_\alpha$
with $\delta_\alpha=1$. Proceeding as above we see that in
${\m H}_0$, the GSO projection with $s=\beta$ breaks the gauge group
$\fs\fu(2)^6$ into $\fu(1)^6$, and the GSO projections with $\beta=\xi_k$
break $\fs\fo(44)$ into $\fs\fo(12)\oplus\fs\fo(16)\oplus\fs\fo(16)$:
\beqno
\chi^i\qu\partial\qu X_\mu
&\mbox{giving}&
\fu(1)^6,\\
\psi^\mu\qu\Phi^i\qu\Phi^j\quad(i,j\in\{1,\ldots,12\})
&\mbox{giving}&
\fs\fo(12),\\
\psi^\mu\qu\varphi^i\qu\varphi^j,\;\psi^\mu\qu\phi^i\qu\phi^j
&\mbox{giving}&
\fs\fo(16)\oplus\fs\fo(16).
\eeqno
In each ${\m H}_{\xi_k}$ the condition \req{massless} yields
$N_L={1\over2},\,N_R=0$. Now let $|\pm\rangle_i$ denote the R ground
states associated to $\qu\varphi^i$ or $\qu\phi^i$, respectively. If
$\kappa=1$ in our tabular below \req{examples}, then no additional
gauge bosons arise from $\m H_{\xi_k}$. But if $\kappa=0$, then 
we get $\psi^\mu\otimes_{i=1}^8 |\pm\rangle_i$ with an even or odd
number of $|+\rangle_i$, depending on the $\eps_k$. In other
words, we get
a $2^7=\mathbf{128}$ spinor
representation of $\fs\fo(16)$, built from the same 
$16$ free fermions which 
give the adjoint $\mathbf{120}$
representation of $\fs\fo(16)$ in ${\m H}_0$. Consistency of the spacetime
theory requires that the gauge bosons must transform in the adjoint of the
gauge algebra; indeed, $\mathbf{128}+\mathbf{120}=\mathbf{248}$ gives
$\mbox{ad}_{\fe_8}$. All in all, since no other sector ${\m H}_\alpha$ 
contains massless states, we  
have the gauge algebra
\beqno
\fu(1)^6\oplus\fs\fo(12)\oplus\fe_8\oplus\fe_8
&& \mb{ if } \kappa=0,\\
\fu(1)^6\oplus\fs\fo(12)\oplus\fs\fo(16)\oplus\fs\fo(16)
&& \mb{ if } \kappa=1.
\eeqno
The case $\kappa=0$ gives
precisely the gauge algebra of the toroidally compactified 
heterotic string theory
with enhanced symmetry $SO(12)\times E_8\times E_8$.
\subsection{Examples of free fermion models}\label{srffm}
It is known that one can use free fermion
models to construct toroidal CFTs with enhanced symmetry. In other words,
for a particular choice of $\m F$ and
the coefficients $C$ in the partition function
\req{pf}, a geometric interpretation will be easy to obtain. 
As explained in Section \ref{gsoorb}, 
adding any basis element to a basis $B$ of $\m F$ results in orbifolding
by a group of type $\Z_2$. In CFT, orbifolding by $\Z_2$ can be reversed
by orbifolding with respect to another group of type $\Z_2$. Hence
omitting basis elements in a free fermion model also amounts to orbifolding. 
Thus all free fermion  models
can be interpreted as orbifolds of a toroidal CFT. These may or may not
be geometric orbifoldings in a given geometric interpretation: For example,
purely geometric group actions can never yield the desired spectra of
Faraggi's semi-realistic free fermion 
models \cite{fny90,fa92}, 
an observation made in  \cite{dofa04}. This was checked in  \cite{dofa04} for the
particular class of geometric orbifolds considered there, and extended to
all geometric orbifolds in our Section \ref{classification}.

In this subsection we will derive geometric interpretations for several
important examples of free fermion models. In particular, we will see
that there exists a free fermion model with geometric interpretation on
the product of three elliptic curves. The corresponding B-field is 
nontrivial, but it is compatible with all geometric orbifoldings 
classified in Section \ref{classification}. In other words, all the 
corresponding geometric orbifoldings can be lifted to the level of CFT, and for 
each of the resulting moduli spaces of orbifold CFTs, there are special points 
giving models which allow a free fermion construction. This also holds for
orbifolds with discrete torsion.
\subsubsection{Free fermion model on the $SO(12)$ torus}
It was already
noted in \cite{muwi86} that the toroidally compactified heterotic string 
with enhanced symmetry $SO(12)\times E_8\times E_8$ can be described using a
free fermion formulation. Let us briefly give the argument in our
language:

First, we identify the partition function of the free fermion model
with basis $B_{tor}$. By what was said in Sections \ref{gsoorb} and
\ref{gagr}, we must choose $\eps^\prime_1=\eps^\prime_2=1$ and
$\kappa=0$ in the table below \req{examples} in order to reproduce the
correct gauge algebra. We introduce $\xi_3:=1+s+\xi_1+\xi_2$ and by
\req{susy} and \req{theta} find
\beqno
Z_{SUSY}(\tau,\qu\tau)
&:=& \inv{2} \left\{ \left( {\theta_3(\tau)\over\eta(\tau)} \right)^4
- \left( {\theta_2(\tau)\over\eta(\tau)} \right)^4
- \left( {\theta_4(\tau)\over\eta(\tau)} \right)^4
- \left( {\theta_1(\tau)\over\eta(\tau)} \right)^4
\right\} \equiv 0,\\
Z(\tau,\qu\tau)
&=& Z_{SUSY}(\tau,\qu\tau)\cdot Z_{Narain}(\tau,\qu\tau),\\
Z_{Narain}(\tau,\qu\tau)
&=& \inv{8} \sum_{\alpha,\beta\in\span_{\F_2}(\xi_1,\xi_2,\xi_3)}
C\!\!\left[\!\!\begin{array}{c}\alpha\\\beta\end{array}\!\!\right]
\prod_{j=9}^{20} 
Z\!\!\left[\!\!\begin{array}{c}\alpha_j\\\beta_j\end{array}\!\!\right]\!\!(\tau)
\prod_{j=21}^{64}
Z\!\!\left[\!\!\begin{array}{c}\alpha_j\\\beta_j\end{array}\!\!\right]\!\!(\qu\tau),
\eeqno
where \req{trafo} allows us to calculate the relevant coefficients $C$, in
particular
$$
\begin{array}[b]{|l||c|c|c|c|}
\hline\vphantom{\sum_{i\over2}^{k\over2}}
\;\;\beta\!\!\!\!&0&\xi_1&\xi_2&\xi_3\\[-0.2em]
\alpha&&&&\\
\hline\hline\vphantom{\sum_{i\over2}^{k\over2}}
0&1&1&1&1\\\hline\vphantom{\sum_{i\over2}^{k\over2}}
\xi_1&1&(-1)^{\eps_1+1}&1&1\\
\hline\vphantom{\sum_{i\over2}^{k\over2}}
\xi_2&1&1&(-1)^{\eps_2+1}&1\\
\hline\vphantom{\sum_{i\over2}^{k\over2}}
\xi_3&1&1&1&(-1)^{\eps_2+\eps_2}\\
\hline
\end{array}\quad\quad
\eps_i\in\{0,1\}.
$$
Note that for all $\alpha,\,\beta\in\span_{\F_2}(\xi_1,\xi_2)$:
$$
C\!\!\left[\!\!\begin{array}{c}\alpha\\\beta+\xi_3\end{array}\!\!\right]
\reeq{trafo}C\!\!\left[\!\!\begin{array}{c}\alpha\\\beta\end{array}\!\!\right]
C\!\!\left[\!\!\begin{array}{c}\alpha\\\xi_3\end{array}\!\!\right]
=C\!\!\left[\!\!\begin{array}{c}\alpha\\\beta\end{array}\!\!\right],
$$
and similarly for all $\alpha,\,\beta\in\{0,\xi_1\}$:
$$
C\!\!\left[\!\!\begin{array}{c}\alpha\\\beta+\xi_2\end{array}\!\!\right]
\reeq{trafo}C\!\!\left[\!\!\begin{array}{c}\alpha\\\beta\end{array}\!\!\right]
C\!\!\left[\!\!\begin{array}{c}\alpha\\\xi_2\end{array}\!\!\right]
=C\!\!\left[\!\!\begin{array}{c}\alpha\\\beta\end{array}\!\!\right].
$$
Together with
$Z\!\!\left[\!\!\begin{array}{c}1\\1\end{array}\!\!\right]\!\!(\qu\tau)
\equiv0$ this implies that a calculation similar to the one performed
in \req{susy} allows to decompose $Z_{Narain}$ as follows:
$$
Z_{Narain}(\tau,\qu\tau)
=\inv2 \sum_{i=1}^4 \left| {\theta_i(\tau)\over\eta(\tau)}\right|^{12}
\cdot \left( \inv2\sum_{i=1}^4 \left(
      {\theta_i(\qu\tau)\over\eta(\qu\tau)}\right)^8
\right)^2.
$$
Since $\inv2\sum_i\theta_i^8$ is the unique modular form of weight $8$
and constant coefficient $1$, it agrees with the theta function $E_4$
of the $E_8$ lattice, and 
\beqn{so(12)part}
Z_{Narain}(\tau,\qu\tau)
&=& Z_{SO(12)}(\tau,\qu\tau)\cdot\left(Z_{E_8}(\qu\tau)\right)^2,\nonumber\\
Z_{SO(12)}(\tau,\qu\tau)
&=&\displaystyle{1\over2} \sum_{i=1}^4 \left|
{\theta_i(\tau)\over\eta(\tau)}\right|^{12}, 
\quad\quad\quad
Z_{E_8}(\qu\tau)
= {E_4(\qu\tau)\over\eta^8(\qu\tau)}.
\eeqn
This was to be expected, since $Z_{Narain}$ is the partition function of
the free fermion model constructed from $12$ left moving fermions
$y^i,w^i,\, i\in\{1,\ldots,6\}$ and $44$ right moving fermions
$\qu y^i,\qu w^i,\, i\in\{1,\ldots,6\},\, \qu\varphi^i,\,\qu\phi^i, \,
i\in\{1,\ldots,16\}$, with spin structures coupled among the 
$y^i,\,w^i,\,\qu y^i,\,\qu w^i$, among the $\qu\varphi^i$, and among the
$\qu\phi^i$. In fact, since we find a $\fu(1)^6_L\oplus\fu(1)^{22}_R$ current
algebra generated by $\nop{y^i w^i},\,\nop{\qu y^i\qu w^i},\, 
\nop{\qu\Phi^{2j-1}\qu\Phi^{2j}}$ ($j>6$), 
this is a toroidal CFT, and the determination of its 
charge lattice will suffice to specify the theory. By the above we already know
that we have an $\fe_8\oplus\fe_8$ gauge symmetry, and thus a
geometric interpretation in terms of a toroidal theory with trivial 
$\fe_8\oplus\fe_8$ bundle on some torus.

It remains to be shown that $Z_{SO(12)}$ is the partition function of the
toroidal CFT at $c=\qu c=6$ with enhanced $SO(12)$ symmetry.
To this end first note that according to the formulas given in 
Appendix \ref{thetafctns},
\beqno
\inv2\left( \theta_3^6(\tau)\qu\theta_3^6(\qu\tau)
+ \theta_4^6(\tau)\qu\theta_4^6(\qu\tau) \right)
&=&
\sum_{\stackrel{\scriptstyle x,y\in\Z^6,}{x-y\in D_6}} 
q^{x^2\over2}\qu q^{y^2\over2},\\
\inv2\left( \theta_2^6(\tau)\qu\theta_2^6(\qu\tau)
+ \theta_1^6(\tau)\qu\theta_1^6(\qu\tau) \right)
&=&
\sum_{\stackrel{\scriptstyle x,y\in\Z^6+{1\over2},}{x-y\in D_6}} 
q^{x^2\over2}\qu q^{y^2\over2},
\eeqno
where we have introduced the root lattice $D_6:=\left\{n\in\Z^6\mid
\sum n_i\equiv0(2)\right\}$ of $SO(12)$, and
${1\over2}\in\left({1\over2}\Z\right)^6$ denotes the vector with all entries
given by ${1\over2}$. Using
$D_6^*=\Z^6\cup\left(\Z^6+{1\over2}\right)$, we find 
$$
Z_{SO(12)}(\tau,\qu\tau)
= {1\over|\eta(\tau)|^{12}}
\sum_{\stackrel{\scriptstyle x,y\in D_6^*,}{x-y\in D_6}} 
q^{x^2\over2}\qu q^{y^2\over2}
= {1\over|\eta(\tau)|^{12}}
\sum_{(x,y)\in \Gamma} 
q^{x^2\over2}\qu q^{y^2\over2}
$$
with
\beq{d6}
\Gamma:=\left\{ (x,y)\in\R^{6,6} \mid x,y\in D_6^*,\, x-y\in
D_6\right\}.
\eeq

The claim now is that $\Gamma$
can be brought into the standard Narain form
\beq{chargegeo}
\Gamma(\Lambda,B)
= \left\{ (p_L,p_R)=\inv{\sqrt2}\left( \mu-B\lambda+\lambda,
\mu-B\lambda-\lambda\right) \mid \lambda\in\Lambda,
\mu\in\Lambda^*\right\}
\eeq
for an appropriate lattice $\Lambda\subset\R^6$ with dual 
$\Lambda^*\subset(\R^6)^*$ (using the standard Euclidean scalar
product on $\R^6$ to view $\Lambda^*\subset\R^6\cong(\R^6)^*$), and for
an appropriate B-field
$B:\Lambda\otimes\R\longrightarrow\Lambda^*\otimes\R$. If for a toroidal
CFT $\m C$ with central charges $c=\qu c=d$
the charge lattice $\Gamma$ can be brought into the form
$\Gamma=\Gamma(\Lambda,B)$ with such $\Lambda,\, B$, then $(\Lambda, B)$
gives a geometric interpretation of $\m C$: The CFT $\m C$ is the nonlinear sigma model
on $\R^d/\Lambda$ with B-field $B$.
Note that any
two B-fields $B,\,B^\prime=B+\delta B$  yield $\Gamma(\Lambda,B)
=\Gamma(\Lambda^\prime,B^\prime)$ iff $\delta B(\Lambda)\subset\Lambda^*$.
For given $\Lambda$ we say that $B,\,B^\prime$ 
are equivalent iff they define the same CFT, i.e.\ iff $\Gamma(\Lambda,B)
=\Gamma(\Lambda^\prime,B^\prime)$.

>From \req{d6} and \req{chargegeo}
we directly read off $\Lambda=\inv{\sqrt2}D_6$, $\Lambda^*=\sqrt2
D_6^*$, such that
$\Lambda^*\subset\Lambda\subset\inv2\Lambda^*$. Since $D_6\subset
D_6^*$, \req{d6} tells us that $\Gamma(\Lambda,B)$ contains all
vectors of type $(p_L,p_R)=(x,0)$ and $(p_L,p_R)=(0,y)$ with $x,\,y\in
D_6$. In other words, $(B-1)\Lambda\subset\Lambda^*$ (or equivalently 
$(B-1)D_6\subset2D_6^*$), which is equivalent to
$(B+1)\Lambda\subset\Lambda^*$ since
$\Lambda\subset\inv2\Lambda^*$. In fact,
$(B-1)D_6\subset2D_6^*$ holds iff all off-diagonal entries of $B$ are
odd, and all such choices of $B$ are equivalent. Without loss of
generality we can therefore take $B=B_\ast$ with
\beq{2torB}
B_\ast=\left(\begin{array}{cccccc}
0&1&1&1&1&1\\
-1&0&1&1&1&1\\
-1&-1&0&1&1&1\\
-1&-1&-1&0&1&1\\
-1&-1&-1&-1&0&1\\
-1&-1&-1&-1&-1&0
\end{array}\right).
\eeq
To show that the free fermion model with basis $B_{tor}$
and $\eps_k^\prime=1,\,\kappa=0$ agrees with the Narain model as
claimed, for instance by using \cite[Thm.\ 3.1]{nawe00},
we still need to identify their W-algebras and charge lattices 
with respect to $\fu(1)^6_L\oplus\fu(1)^6_R$. Since the theory is
left-right symmetric, we can focus on the left-handed degrees of 
freedom. With
\beqno
j_k&:=&i\nop{y^kw^k},\quad k\in\{1,\ldots,6\},\\
x^k&:=&\inv{\sqrt2}\left(y^k+i w^k\right),\quad\quad
(x^k)^*:=\inv{\sqrt2}\left(y^k-i w^k\right),\quad
 k\in\{1,\ldots,6\},
\eeqno
in addition to the $j_k$ which generate $\fu(1)^6$
we find $60$ further $(1,0)$ fields in the free fermion model:
$$
i\nop{x^k x^l}=-i\nop{x^l x^k},\quad\quad
i\nop{x^k(x^l)^*},\quad\quad
i\nop{(x^k)^*(x^l)^*}=-i\nop{(x^l)^*(x^k)^*}\quad\quad (k\neq l),
$$
with charges with respect to $\vec\jmath=(j_1,\ldots,j_6)$ given by
$$
e_k+e_l,\quad
e_k-e_l,\quad
-e_k-e_l,
$$
where the $e_i$ denote the standard basis vectors in 
$\R^6$. Hence we can identify these $(1,0)$ fields with the holomorphic
vertex operators $V_{(p_L,0)}$ of the respective charges $(p_L,0)$.
One checks that this identification is compatible with the OPE, i.e.\
the free fermion model and our toroidal CFT share the same 
W-algebra, with zero mode algebra of the generators given by $\fs\fo(12)$.
By a similar analysis one identifies all $V_{(p_L,p_R)}$ for
$(p_L,p_R)\in\Gamma_{Narain}$ with fields in the free fermion model:
Since we have already dealt with the $(1,0)$ fields, and since in our model
$(B\pm1)\Lambda\subset\Lambda^*$, it suffices to identify the
$V_{(p_L,p_R)}$ with $p_L=p_R=\inv{\sqrt2}\mu,\,\mu\in\Lambda^*$. 
Now
$$
i\nop{x^k\qu x^l}\longmapsto V_{(e_k,e_l)}\quad
\mb{ and }\quad
\prod_{i=1}^6 |\delta_i\rangle_i\qu{|\delta_i\rangle_i}
\longmapsto V_{({1\over2}\sum_i\delta_i e_i,{1\over2}\sum_i\delta_i e_i)},
\;\;\delta_i\in\{\pm\}
$$
gives the desired identification.
\subsubsection{Free fermion model on the square torus}\label{squaretorus}
%
In the previous subsection, we have argued that a free fermion model 
with basis $B_{tor}$ for $\m F$ 
yields a conformal field theory with geometric interpretation
on the torus $\R^6/\Lambda$ with $\Lambda = {1\over\sqrt2} D_6$. The lattice 
$\Lambda^\prime=\sqrt2\Z^6$ is a sublattice of $\Lambda$ of index $2^5$. Correspondingly,
for the dual lattices we find that $(\Lambda^\prime)^\ast$ is 
generated by $\Lambda^\ast$ and the multiples ${1\over\sqrt2} e_i$ of
the first five standard basis vectors $e_i,\, i\in\{1,\ldots,5\}$. 
This is evidence for the fact that there
is also a free fermion model with geometric interpretation on the square torus
$T^6=\R^6/{\sqrt2}\Z^6$: It should arise by orbifolding with respect to a group of type  
$(\Z_2)^5$ from the toroidal free fermion model on the $SO(12)$ torus.

Indeed, with the same techniques as in the previous section, one shows:
Consider the free fermion model with basis 
$B_\Box=\{s,\zeta_1,\ldots,\zeta_5,\xi_1,\xi_2,\xi_3\}$,
where $\xi_3=1+s+\xi_1+\xi_2$ as before, and
$\zeta_i$ with $i\in\{1,\ldots,5\}$ is the vector which has an entry $1$ corresponding to
the fermions $y^i,\,w^i,\,\qu y^i,\,\qu w^i$ and entries $0$ otherwise. For the coefficients $C$, 
for all $\beta\in B_\Box$, we set
$C\!\!\left[\!\!\begin{array}{c} s\\\beta\end{array}\!\!\right]:=-1$,
while for $\alpha,\beta\in B_\Box-\{ s\}$, we set
$C\!\!\left[\!\!\begin{array}{c}\alpha\\\beta\end{array}\!\!\right]:=1$.
The resulting free fermion model has geometric interpretation on the square torus
$T^6=\R^6/{\sqrt2}\Z^6$ with the same B-field $B_\ast$ as for the previous
toroidal model, c.f.\ \req{2torB}. 

This is an important observation with respect to our classification in 
Section \ref{classification}. It implies that for all the orbifolds $X/G$ 
given there, $X \cong T^6$ with the complex
structure of a product $E_1 \times E_2 \times E_3$ of three elliptic curves,
a free fermion model exists which has geometric interpretation on $X/G$, provided
that the action of $G$ is compatible with the B-field $B_\ast$ given in \req{2torB}. 
Compatibility here means that for every $g\in G$, the $g$-conjugate B-field is
equivalent to $B_\ast$, which indeed is the case for all groups $G$
discussed in Section \ref{classification}.

Note that the above orbifolding by
$(\Z_2)^5$  is not described in terms of a geometric orbifolding: 
While a geometric orbifolding would have to lead to a model with geometric
orbifold interpretation on a quotient of $\R^6/\Lambda$, 
the geometric interpretation $(\Lambda^\prime,B)=(\sqrt2\Z^6, B_\ast)$
of the orbifold CFT yields an unbranched cover 
$T^6=\R^6/\sqrt2\Z^6$ of the geometric interpretation 
$(\Lambda,B)=({1\over\sqrt2}D_6, B_\ast)$ of the original theory on $\R^6/{1\over\sqrt2}D_6$. 
The reverse of this orbifolding,
obtained in the free fermion language by omitting the basis vectors 
$\zeta_i,\,i\in\{1,\ldots,5\}$ from the basis $B_\Box$, is a geometric orbifolding
of type $(\Z_2)^5$ by shifts, and it 
yields $\R^6/{1\over\sqrt2}D_6 = T^6/(\Z_2)^5$.
\subsubsection{The NAHE model}\label{NAHE}
As an example of orbifolding by a group which does 
not act as shift orbifold on the torus, we consider the free fermion 
model with basis 
$B_{NAHE^+}=\{1,s,\xi_1,\xi_2,g_1,g_2\}$.  
This is the geometric part of what Faraggi calls 
the extended NAHE set \cite{fgkp87,inq87,aehn89,fny90,fa92,fana93}, and in
the notation of \cite{dofa04} one has
$g_k=b_k+s+\xi_2$  and $g_3=g_1+g_2+1+\xi_1$. The $8$ Dirac
fermions $\qu\varphi^i$ are renamed into 
$\qu\psi^1,\ldots,\qu\psi^5,\,\qu\eta^1,\,\qu\eta^2,\,\qu\eta^3$.
Omitting untwisted fermions, for the additional basis vectors we set
\beq{gkaction}
\begin{array}[b]{|l||c|c|c||c|c||c|c||c|c||}
\hline
&\chi^1,\,\chi^2,&\chi^3,\,\chi^4&\chi^5,\,\chi^6
&y^1,\,y^2,&w^1,\,w^2,
&y^3,\,y^4,&w^3,\,w^4,
&y^5,\,y^6,&w^5,\,w^6,
\\
&\qu\eta^1&\qu\eta^2&\qu\eta^3
&\qu y^1,\,\qu y^2&\qu w^1,\,\qu w^2
&\qu y^3,\,\qu y^4&\qu w^3,\,\qu w^4
&\qu y^5,\,\qu y^6&\qu w^5,\,\qu w^6
\vphantom{\sum_{i\over2}}
\\[-0.2em]
\hline\hline\vphantom{\sum_{i\over2}^{k\over2}}
g_1&0&1&1&
0&0&1&0&1&0
\\\hline\vphantom{\sum_{i\over2}^{k\over2}}
g_2&1&0&1&
1&0&0&0&0&1
\\\hline\vphantom{\sum_{i\over2}^{k\over2}}
g_3&1&1&0&
0&1&0&1&0&0
\\\hline
\end{array}.
\eeq
The geometric action on the $SO(12)$ torus model with pre-Hilbert
space
$\m H_{SO(12)}$ built on the $y^i,\,w^i;\,\qu y^i,\,\qu w^i$
is left-right symmetric.
Translating  into the fundamental fields of the toroidal theory
we get
$$
\begin{array}{lllrcrl}
g_1:&j_k\mapsto-j_k&\mb{for }\;\;k\in\{3,\ldots,6\},
&x^k&\longleftrightarrow&-(x^k)^*
&\mb{for}\;\;k\in\{3,\ldots,6\},\\[0.5em]
g_2:&j_k\mapsto-j_k&\mb{for }\;\;k\in\{1,2,5,6\},
&x^k&\longleftrightarrow&-(x^k)^*
&\mb{for}\;\;k\in\{1,2\},\\[0.5em]
&&&x^k&\longleftrightarrow&(x^k)^*
&\mb{for}\;\;k\in\{5,6\},\\[0.5em]
g_3:&j_k\mapsto-j_k&\mb{for }\;\;k\in\{1,\ldots,4\},
&x^k&\longleftrightarrow&(x^k)^*
&\mb{for}\;\;k\in\{1,\ldots,4\},
\end{array}
$$
and analogously for the right-handed fields.
In geometric language with real coordinates $v_1,\ldots,v_6$ this
corresponds to 
\beq{gkgeoaction}
\left.
\begin{array}{rcl}
\displaystyle 
g_1:\quad v_k\mapsto-v_k\>\>\mb{for }\;\;k\in\{3,\ldots,6\},
\mb{ up to a shift on the charge}\e
\>\>\mb{ lattice by}
\;\;\delta=\inv2\left(1,1,0,0,0,0;1,1,0,0,0,0\right),\e
g_2:\quad v_k\mapsto-v_k\>\>\mb{for }\;\;k\in\{1,2,5,6\},
\mb{ up to a shift on the charge}\e
\>\>\mb{ lattice by}
\;\;\delta=\inv2\left(1,1,0,0,0,0;1,1,0,0,0,0\right),\e
g_3:\quad v_k\mapsto-v_k\>\>\mb{for }\;\;k\in\{1,\ldots,4\},
\mb{ i.e. } g_3=g_1\circ g_2.
\end{array}
\right\}
\eeq
The claim is that the shifts involved in $g_1,\,g_2$ can be
ignored, i.e.\ that $g_1,\,g_2,\,g_3$ act geometrically
as the three non-trivial
elements of the Kleinian $(\Z_2)^2$ twist group $T_0$.

To see this, let us assume that the $g_k$ act as claimed in the geometric
interpretation and derive \req{gkaction} from this assumption. By the
above, we only need to confirm the choices between placing the $1$'s
in the $\{y^k\}$ instead of the $\{w^k\}$ columns in \req{gkaction}
for $g_1,\,g_2$. 
First note that for the construction of the untwisted sector of the
orbifold these choices are irrelevant. Namely, the additional sign in 
$x^k\leftrightarrow-(x^k)^*$ merely results in a choice of, say, 
$\nop{x^1x^3}-\nop{x^1(x^3)^*}$ instead of
$\nop{x^1x^3}+\nop{x^1(x^3)^*}$ as invariant field under $g_1$, with
no consequence on the OPE. In accord with this, all the contributions
to the partition function
\beq{frfinv}
\orbox{g_k}{1}(\tau)
=\tr[\m H_{SO(12)}] \left(g_k q^{L_0-6/24}\qu q^{\qu L_0-6/24}\right)
=\inv2\left\{ \left|{\theta_4\over\eta}\right|^4
\left|{\theta_3\over\eta}\right|^8
+\left|{\theta_3\over\eta}\right|^4
\left|{\theta_4\over\eta}\right|^8\right\}
\eeq
for $k\in\{1,2,3\}$ agree,
where the factors raised to the fourth power in each summand come from
the action of the twist on four of the real fermions.
Similarly the traces over the full 
$g_k$ twisted sectors of the orbifold are
\beq{frfpart}
\orbox{1}{g_k}(\tau)=\orbox{g_k}{1}(-\inv\tau)
=\inv2\left\{ \left|{\theta_2\over\eta}\right|^4
\left|{\theta_3\over\eta}\right|^8
+\left|{\theta_3\over\eta}\right|^4
\left|{\theta_2\over\eta}\right|^8\right\},
\eeq
where it again should be kept in mind that the factors 
raised to the fourth power in each summand come from
the action of the twist on four of the fermions.
The choice of placing the $1$'s in the 
$\{y^k\}$ instead of the $\{w^k\}$ columns in \req{gkaction}
hence only enters into the encoding of the $g_i$ action on the $g_k$
twisted sector with $i\neq k$.

In terms of the geometric interpretation of the toroidal theory 
equivalently to \req{frfinv} and \req{frfpart} we
write
\beqn{twist}
\orbox{g_k}{1}(\tau)
&=&\ds
\left|{2\eta\over\theta_2}\right|^4
\cdot\inv2\left\{ \left|{\theta_3\over\eta}\right|^4
+\left|{\theta_4\over\eta}\right|^4\right\},\nonumber\\
\orbox{1}{g_k}(\tau)
&=&\ds
\left|{2\eta\over\theta_4}\right|^4
\cdot\inv2\left\{ \left|{\theta_3\over\eta}\right|^4
+\left|{\theta_2\over\eta}\right|^4\right\},
\eeqn
where the first factor in each case accounts for the contributions
from the twisted states in four real coordinate directions, whereas
the second factor comes from the trace over states left invariant
by $g_k$. Hence in the usual $(\Z_2)^2$ orbifold, when 
$g_i$ with
$i\neq k$ acts on the $g_k$ twisted sector, it must leave a factor
$\left|{2\eta\over\theta_4}\right|^2$ invariant, and act by the usual $\Z_2$
twist on a second factor $\left|{2\eta\over\theta_4}\right|^2$
transforming it into $\left|{2\eta\over\theta_3}\right|^2$, while 
it introduces the usual factor $\left|{\theta_3\theta_4\over\eta^2}\right|^2
=\left|{2\eta\over\theta_2}\right|^2$ for a twisted sector in the
directions which are left invariant by $g_k$ but not by $g_i$.
All in all we get
\beq{iktwist}
\mb{for }i\neq k:\quad\quad
\orbox{g_i}{g_k}(\tau)
= \left|{2\eta\over\theta_4}\right|^2
\left|{2\eta\over\theta_3}\right|^2
\left|{2\eta\over\theta_2}\right|^2
\reeq{vanish} 2^4.
\eeq
Let us now translate the $g_i$ action on the $g_k$ twisted sector back
into the language of the free fermion model \req{frfpart}. We already
know that in \req{frfpart} a global factor 
$\left|{\theta_2\theta_3\over\eta^2}\right|^2$ must remain invariant,
coming from the directions twisted by $g_k$ but not by $g_i$.
A factor $\left|{\theta_3\over\eta}\right|^4$ in the first
summand is transformed into
$\left|{\theta_3\theta_4\over\eta^2}\right|^2$, and 
a factor $\left|{\theta_2\over\eta}\right|^4$ in the second
summand is transformed into
$\left|{\theta_2\theta_1\over\eta^2}\right|^2=0$,
each coming from directions twisted by $g_i$ but not by $g_k$. 
Since by
\req{iktwist} the final
result of the transformation must be
$
\left|{\theta_2\theta_3\over\eta^2}\right|^2
\cdot\left|{\theta_2\theta_4\over\eta^2}\right|^2
\cdot\left|{\theta_3\theta_4\over\eta^2}\right|^2,
$
we find the remaining factor coming from directions twisted by
both $g_i$ and $g_k$, namely
$\left|{\theta_2\theta_3\over\eta^2}\right|^2$ in the first summand
is transformed into
$\left|{\theta_2\theta_4\over\eta^2}\right|^2$. Hence the twist is
applied to fermions previously yielding a $\theta_3$ contribution, in
other words to fermions which had been untwisted so far. Altogether
this indeed leads to the data listed in \req{gkaction} for the 
$y^i,\,w^i;\,\qu y^i,\,\qu w^i$.

In the construction of a semi-realistic free fermion model 
\cite{fny90,inq87}, the authors also use 
three further $\Z_2$ actions, $\alpha,\,\beta,\,\gamma$, where again
we only list fermions that are in fact twisted:
$$
\begin{array}[b]{|l||c|c||c|c||c|c||c||c||c||}
\hline
&y^1,\,y^5,&
&y^2,\,y^4,&
&y^3,\,y^6,&&&&
\\
&w^1,\,w^5,&\qu y^1,
&w^2,\,w^4&\qu w^2,
&w^3,\,w^6&\qu y^3,&&&
\\
&\qu w^1,\,\qu y^5,&\qu w^5
&\qu y^2,\,\qu y^4&\qu w^4
&\qu w^3,\,\qu w^6&\qu y^6
&\qu\psi^{1,\ldots,5}&\qu\eta^i&\qu\phi^{1,\ldots,8}
\vphantom{\sum_{i\over2}}
\\[-0.2em]
\hline\hline\vphantom{\sum_{i\over2}^{k\over2}}
\alpha&0&1&0&1&1&0
&1\,1\,1\,0\,0&0&1\,1\,1\,1\,0\,0\,0\,0
\\\hline\vphantom{\sum_{i\over2}^{k\over2}}
\beta&1&0&0&1&0&1
&1\,1\,1\,0\,0&0&1\,1\,1\,1\,0\,0\,0\,0
\\\hline\vphantom{\sum_{i\over2}^{k\over2}}
\gamma&0&1&1&0&0&1
&\inv2\,\inv2\,\inv2\,\inv2\,\inv2&0&\inv2\,0\,1\,1\,\inv2\,\inv2\,\inv2\,0
\\\hline
\end{array}.
$$
Note that $\alpha,\,\beta,\,\gamma$ share the property of leaving all
the $j_k$ invariant and multiplying all the $\qu\jmath_k$ by
$-1$. This is hard to interpret geometrically, since the left-right
coupling of the local coordinate functions corresponding to the pairs 
$y^i+\qu y^i, w^i+\qu w^i$ is broken.
The difference in sign between the action on the holomorphic $j_k$ and the
antiholomorphic $\qu\jmath_k$ is reminiscent of
some type of mirror symmmetry.
The actions of $\alpha\beta,\,\beta\gamma,\,\gamma\alpha$
on the $y^i,\,w^i;\,\qu y^i,\,\qu w^i$, however, have a geometric
interpretation:
$$
\begin{array}[b]{|l||c|c|c|}
\hline
&y^1,\,y^5,\,w^1,\,\qu y^1,
&y^2,\,y^4,\,w^2,\,\qu w^2,
&y^3,\,w^3,\,w^6,\,\qu y^3,
\\
&w^5,\,\qu y^5,\,\qu w^1,\,\qu w^5
&\qu y^2,\,\qu y^4,\,w^4,\,\qu w^4
&y^6,\,\qu w^3,\,\qu w^6,\,\qu y^6
\vphantom{\sum_{i\over2}}
\\[-0.2em]
\hline\hline\vphantom{\sum_{i\over2}^{k\over2}}
\alpha\beta&1&0&1\\\hline\vphantom{\sum_{i\over2}^{k\over2}}
\beta\gamma&1&1&0\\\hline\vphantom{\sum_{i\over2}^{k\over2}}
\gamma\alpha&0&1&1\\\hline
\end{array}.
$$
Each of these actions is left-right symmetric and acts trivially on
all $j_k,\,\qu\jmath_k$. As such, they are shift orbifolds, namely by
\beqno
\alpha\beta: && \inv2\left(0,1,0,1,0,0;0,1,0,1,0,0\right),\\
\beta\gamma: && \inv2\left(0,0,1,0,0,1;0,0,1,0,0,1\right),\\
\gamma\alpha:&& \inv2\left(1,0,0,0,1,0;1,0,0,0,1,0\right).
\eeqno
Faraggi shows that his model is a three generation model \cite{fa92},
which is a necessary requirement for a theory to be viewed
as ``semi-realistic''. It is natural to ask whether there exists 
an underlying geometric orbifold with Hodge numbers 
$(h^{1,1}, h^{2,1})$ yielding three generations
$3=h^{1,1}- h^{2,1}$. In \cite{dofa04}, this question was
answered to the negative, however without a complete classification
of all possible orbifolds. Our classfication, summarized in Table 1
of Section \ref{tables}, completes this task, and again answers
the question to the negative. The numbers of generations that
can be produced by purely geometric methods, according to the
results of Section \ref{tables}, are $48,\,24,\,12,\,6$, or $0$.
It is interesting that precisely the number $3$ is lacking in
this list.
\section{Special models within our classification}\label{special}
In this section, we discuss some special cases of the orbifolds that we have classified
in Section \ref{classification}. More precisely, we identify some of the resulting 
Calabi-Yau threefolds as degenerate cases of so-called Borcea-Voisin threefolds and
Schoen threefolds or their orbifolds. All these particular Calabi-Yau threefolds
have been widely discussed in the literature, either in relation to mirror symmetry
or to model building in heterotic string theory. 
Since the results of Section \ref{ffm} in particular
imply that for every Calabi-Yau threefold listed in Section \ref{tables} there exists
a free fermion model of an associated CFT, we  automatically obtain free fermion constructions
for theories associated to certain Borcea-Voisin threefolds, Schoen threefolds, and their orbifolds.
This may eventually yield further insight into the geometry of these threefolds, 
and it may simplify some of the
existing string theory constructions, since free fermion models are constructed using very simple
mathematical tools.
\subsection{The Vafa-Witten and NAHE models}
As was briefly mentioned at the end of our discussion of Table 1 in
Section \ref{tables}, our model
$(0-1)$ agrees with the $(\Z_2)^2$ orbifold which was extensively studied
by Vafa and Witten in their seminal work \cite{vawi95} on discrete torsion
and mirror symmetry. Since the Vafa-Witten model is indeed obtained as
orbifold of the product of three elliptic curves by the group $T_0$ of ordinary
twists, agreement with our model $(0-1)$ is immediate.

Let us now discuss the two models $(1-1)$ and $(2-9)$ in our list, both of which have
Hodge numbers $(27,3)$. They are not equivalent as topological spaces, since
they can be distinguished by their fundamental groups $C=\Z_2$ and $0$, respectively.
However, there seems to have been some confusion between these two models,
which we now wish to lift. 
Clearly, $(1-1)$ is obtained as $\Z_2$-orbifold of the Vafa-Witten
model. On the other hand, we claim that $(2-9)$ agrees with the Calabi-Yau threefold $Y$
which is obtained by orbifolding an $SO(12)$ torus by the orbifolding group $T_0$. 
This follows using the ideas described at the end of Section \ref{squaretorus}:
The $SO(12)$ torus can be obtained from the product $X$ of three elliptic curves 
by a shift orbifold using the group 
$\widetilde G_S:=(\Z_2)^5$ with 
generators
$$
(\tau,0,0), (0,\tau,0), (0,0,\tau), (0,1,1), (1,0,1).
$$
Hence $Y$ is topologically equivalent to $X/\widetilde G$, where 
$\widetilde G=\widetilde G_S\times T_0$. However, the group $\widetilde G$ is 
redundant, since shifts by the first three vectors listed
above are redundant. Hence $Y$ is also topologically equivalent
to $X/G$ with $G$ generated by $T_0$ and the shifts $(0,1,1), (1,0,1)$.
This is precisely our model $(2-9)$. 

It now follows that the free fermion model 
with basis $B_{NAHE^+}$ discussed in Section \ref{NAHE} gives a CFT
with geometric interpretation on our threefold $(2-9)$: In Section \ref{NAHE}
we have described this free fermion model as a $(\Z_2)^2$-orbifold of the toroidal 
model on the $SO(12)$ torus, and \req{gkgeoaction} identifies the
relevant action of $(\Z_2)^2$ with $T_0$. By the above, this gives a geometric
interpretation on $(2-9)$.
It also means that the NAHE free fermion model with basis $B_{NAHE^+}$
does \textsl{not} have a geometric interpretation on the 
$\Z_2$-shift orbifold $(1-1)$ of the Vafa-Witten model, as is sometimes
claimed. Using the techniques described so far, one also checks
that the free fermion model with basis $B_{NAHE^+}\cup\{\alpha\beta,\beta\gamma,
\gamma\alpha\}$ (see Section \ref{NAHE} for notations) has geometric interpretation
on our  Calabi-Yau threefold $(4-1)$ with Hodge numbers $(15,3)$. Faraggi, on the other
hand, constructs a semi-realistic free fermion model with chiral spectrum $(6,3)$ \cite{fny90,fa92}. 
As can be seen from our classification in Section \ref{tables}, there is no geometric
orbifold of the appropriate type with these Hodge numbers.
\subsection{Borcea-Voisin threefolds}\label{BV}
Within our list of orbifolds tabulated in Section \ref{tables}, 
there are several examples of Borcea-Voisin threefolds \cite{bo97,vo93}. 
Namely, let $BV(r,a,\delta)$ denote 
a connected component of the moduli space of Borcea-Voisin 
threefolds obtained by a $\Z_2$ orbifolding procedure from the product
of a $K3$-surface and an elliptic curve, $(K3\times E_3)/(\iota,-1)$,
where $\iota$ acts as antisymplectic automorphism on $K3$. Here, $(r,a,\delta)\in\N^3$ 
are the parameters from Nikulin's classification
of $K3$ surfaces with such automorphisms \cite{ni79}. 
These parameters uniquely specify
the topological invariants of each element in $BV(r,a,\delta)$, and 
there are precisely $75$ possible triples $(r,a,\delta)$. One finds that the
Hodge numbers of the resulting Borcea-Voisin threefolds are
$$
h^{1,1}=5+3r-2a, \quad h^{2,1}=65-3r-2a,
$$
except for $(r,a,\delta)=(10,10,0)$ where $h^{1,1}=h^{2,1}=11$,
see e.g.\ \cite{bo97,vo93}.
A related set of invariants describes the components of the fixed locus of the involution 
$\iota$. In all cases except $(10,10,0)$ and $(10,8,0)$, this set consists of $k+1$ curves, 
one of which has genus $g$ while the others are rational. In the exceptional case $(10,10,0)$
the fixed locus is empty, while in case $(10,8,0)$ it consists of two elliptic curves.
In the remaining cases, these invariants are related to $r,a$ by:
$$
2g = 22-r-a,\quad         2k=r-a.
$$


We claim that seven of the orbifolds listed in Section \ref{tables} are among the 
Borcea-Voisin families of threefolds:
\begin{eqnarray*}
(0-1) &\in BV(18,4,0), & (h^{1,1},h^{2,1})=(51,3),\\
(0-2) &\in BV(10,8,0),  & (h^{1,1},h^{2,1})=(19,19),\\
(0-3),\,(1-10) &\in BV(10,10,0), & (h^{1,1},h^{2,1})=(11,11),\\
(1-6) &\in BV(14,8,1), & (h^{1,1},h^{2,1})=(31,7),\\
(1-8) &\in BV(10,10,1), & (h^{1,1},h^{2,1})=(15,15),\\
(2-13) &\in BV(12,10,1), & (h^{1,1},h^{2,1})=(21,9).
\end{eqnarray*}
In general, any automorphism of a two-dimensional abelian variety that
commutes with the $(-1)$ involution permutes its $16$ fixed points and induces an
isomorphism between the tangent spaces at corresponding points. It hence
lifts to an automorphism of the $K3$ surface obtained by resolving the Kummer 
surface. The
symplectic form must be mapped to some multiple of itself, and that multiple can
be evaluated at any point of the resulting $K3$ surface. We may therefore safely 
ignore the fixed points and work on the torus $E_1 \times E_2 \times E_3$.

We write these seven quotients in the form 
$(E_1 \times K3)/(-1,\iota)$, where $-1$ sends $x \mapsto -x$ while $\iota$ is
induced (as above) from an involution (still denoted $\iota$) of $E_2 \times E_3$. 
The twist part of $\iota$ will always be $(y,z) \mapsto (y,-z)$.

In each case we write:
\begin{itemize}
\item
The group acting on $E_1 \times E_2 \times E_3$ (in a couple of cases we need a
permutation of what we have in Section \ref{classification}).
\item
The subgroup $G_0$ fixing $E_1$ and acting only on $E_2 \times E_3$.
\item
The involution $\iota$ on $E_2 \times E_3$.
\item
The fixed curves of $\iota$ and its composites with $G_0$ in $E_2 \times E_3$ and
their image in the $K3$ surface, i.e.\ mod $G_0$, that is the ramification curve of 
the $K3$ involution.
\item
The invariants $g,k$ when they make sense (i.e.\ except in cases $(10,10,0)$ and $(10,8,0)$,
when the fixed locus is empty or two elliptic curves, respectively), and $(r,a,\delta)$.
\end{itemize}
$$
\begin{array}{|c|c|c|c|l|l|}
\hline
\mb{model}&\mb{group}&G_0&\iota&\mb{Fix}(\iota\cdot G_0)&(k,g),\\
&&&&\quad\to\mb{ramif.\ curve}&
\;\;(r,a,\delta)\\
\hline\hline
(0-1)&&&& 8 \mb{ elliptics:}
&\\
&(0+,0-,0-),&(0-,0-)& (0+,0-)
&\{2y=0\} \cup \{2z=0\}&(0,7),\\
&(0-,0+,0-)& & 
&\quad\to8 \mb{ rationals}&\;\; (18, 4, 0)\\
\hline
(0-2)&(0+,0-,0-),&(0-,0-)&(0+,1-)&4 \mb{ elliptics: }\{2z=1\}&\\
&(0-,0+,1-)& &&\quad\to 2 \mb{ elliptics}&\;\;(10,8,0)\\
\hline
(0-3)&(0+,0-,0-),&&& \mb{empty} &\\ 
&(0-,1+,1-) &(0-,0-)&(1+,1-)&\quad\to \mb{ empty}&\;\;(10,10,0)\\
\hline
(1-6)&(0+,0-,0-),&&& 8 \mb{ elliptics:} &\\ 
&(0-,0+,0-), &(0-,0-),&(0+,0-)&\{2y=0\} \cup \{2z=0\}&(0,3),\\
&(0,t,t) &(t,t)&&\quad\to 4 \mb{ rationals}&\;\;(14, 8, 1)\\
\hline
(1-8)&(0+,0-,1-),&&&  &\\ 
&(0-,0+,0-), &(0-,1-),&(0+,0-)&4 \mb{ elliptics: }\{2z=0\}&(1,0),\\
&(0,t,t) &(t,t)&&\quad\to 1 \mb{ elliptic}&\;\;(10, 10, 1)\\
\hline
(1-10)&(0+,0-,0-),&&&  &\\ 
&(0-,1+,1-), &(0-,0-),&(1+,1-)&\mb{empty}&\\
&(0,t,t) &(t,t)&&\quad\to \mb{empty}&\;\;(10,10,0)\\
\hline
(2-13)&(0+,0-,0-),&&&&\\ 
&(0-,0+,0-),&(0-,0-),&& 8 \mb{ elliptics:} &\\ 
&(0,1,1),  &(1,1),&(0+,0-)&\{2y=0\} \cup \{2z=0\}&(0,1),\\
&(0,t,t) &(t,t)&&\quad\to 2 \mb{ rationals}&\;\;(12, 10, 1)\\
\hline 
\end{array}
$$

The above argument shows that $(1-10)$ is in the same family as $(0-3)$. 
More precisely, these are two distinct three-parameter subfamilies of the 
eleven dimensional family of Borcea-Voisin threefolds of type 
$(10,10,0)$. In each case, the three
parameters arise as the modulus of the elliptic curve plus two moduli for
Kummer-like $K3$ surfaces, but these are two different two-parameter families of
the latter.

As to the determination of the invariants $(r,a,\delta)$,
the above calculations give us the fixed divisor in the orbifolding, hence by
standard formulas also $r$ and $a$. To obtain $\delta$,
in case $(0-1)$ we check explicitly that the
class of the ramification divisor is even, basically because it has even
multiplicity (namely, two) at each of the 16 blown up points. 
It follows that $\delta=0$ in this case.
In all
other cases $\delta$ is uniquely determined, either because only one
possibility occurs in Nikulin's list, or because the fixed divisor is either empty or 
it consists of two
elliptic curves, which means that these yield cases $(10,10,0)$ and $(10,8,0)$,
respectively.

%
It is curious that all the examples of Borcea-Voisin threefolds which
occur in our list 
either have Hodge numbers $h^{1,1}=h^{2,1}$ or 
do not have Borcea-Voisin mirror partners since they have parameters $(r,a,\delta)$
where $(20-r,a,\delta)$ does not 
belong to the list of $75$ possible triples found by Nikulin \cite{ni79}. 
Again, the most prominent example of this type is the model 
$(0-1) \in BV(18,4,0)$ discussed by Vafa and Witten in \cite{vawi95}. 
For each of these models, it seems that discrete torsion allows the 
construction of a mirror partner. Using our results, one even has
free fermion constructions for examples of CFTs associated to these 
``exceptional'' Borcea-Voisin threefolds.
\subsection{The Schoen threefold and its descendants}\label{schoen}
We remark that
our orbifold $(0-2)$, with Hodge numbers $(19,19)$, can be identified with
Schoen's threefold \cite{sch88}.  This may be of importance for the study
of semi-realistic heterotic string theories, as we shall explain below.
Let us first argue why $(0-2)$ does indeed agree with Schoen's threefold
\cite{sch88} which
is obtained as the fiber product over
$\P^1$ of two rational elliptic surfaces $S_1,\,S_2$.

To this end note first that  Schoen's threefold has Hodge numbers $(19,19)$
in agreement with our claim. Namely,
the complex structure of each rational elliptic surface depends on $8$
(complex) parameters, and three more parameters are needed to fix an
isomorphism between the two $\P^1$ bases, resulting in $8+8+3=19$
parameters in all. 
We claim that our orbifolds $(0-2)$ form a $3$ dimensional subfamily of the family of
Schoen threefolds. The rational elliptic surface, which generically has $12$
degenerate fibers of type $I_1$, specializes here to an isotrivial one, having two
degenerate fibers of type $I_0^*$
and all other fibers having a fixed value of the $j$-invariant.
These surfaces depend on a single complex parameter, the fixed value of
$j$. Since these surfaces have automorphisms acting non trivially on the
$\P^1$ bases, we get only one additional parameter for matching the
bases, for a total of $1+1+1=3$ parameters, accounting for the moduli of
our three elliptic curves $E_i$.

To finally identify our threefolds of type $(0-2)$ with Schoen's threefold,
note that our orbifolds can be written in the form:
$$ 
Y = S_1 \times_{\P^1} S_2, 
$$
where in the obvious notation:
\begin{eqnarray*}
S_1 &:=& (E_1 \times E_3) / \langle (0+,0-),(0-,1-)\rangle,\\
S_2 &:=& (E_2 \times E_3) / \langle (0-,0-),(0+,1-)\rangle,\\
\P^1 &:=& E_3 / \langle (0-),(1-)\rangle = (E_3 /\langle (1+)\rangle )/\langle(0-)\rangle. 
\end{eqnarray*}
Each $S_i$ maps to this $\P^1$, with constant fiber $E_i$ except over two
points of $\P^1$ where the fiber degenerates.

The various quotients of our orbifold $(0-2)$ can be similarly identified
with quotients of special cases of the Schoen threefolds. Of greatest
immediate interest is orbifold $(1-3)$. This was studied in \cite{dopw00} in an
attempt to construct heterotic string compactifications with the low
energy spectrum of the Standard Model of particle physics. This attempt
succeeded through the construction of a different heterotic vector bundle
on the same threefold, in \cite{bodo06}, some of whose physical properties were
further investigated in \cite{bcd06}. Note that our identification of $(1-3)$
with the threefold used in these works implies that free fermion constructions
may suffice to construct the associated string theories. This would dramatically
simplify the rather technical approach of \cite{dopw00,bodo06}.

All free group actions on Schoen threefolds were analyzed in \cite{bodo07}, where
they are tabulated in Table 11. The last two, with fundamental group
$\Z_2$, correspond to our models $(1-3)$ and $(1-7)$. In \cite{bodo07} they are
distinguished by the invariants $m=2$ and $m=1$, respectively. The two
quotients with fundamental group $(\Z_2)^2$ correspond to our
models $(2-5)$ and $(2-14)$, corresponding again to $m=2$ and $m=1$, respectively.

Let us argue that the invariant $m$ in Table 11 of \cite{bodo07} 
can indeed be used to distinguish our families. Let $Y$ be a Schoen quotient, and 
$\pi\colon \widetilde{Y} \to Y$ its universal cover, of degree $n$.
The Schoen quotient $Y$ has a fibration $f\colon Y \to \P^1$.
The composition $\widetilde f := f \circ\pi$ is the original abelian surface fibration of the Schoen threefold $\widetilde{Y}$.  
The generic fiber $A = E_1 \times E_2$ of $\widetilde f$ 
is the product of two elliptic curves $E_1,E_2$, and the generic fiber of $f$ is its quotient by a finite subgroup. The invariant $m$ is defined so that the size of this subgroup is $n/m$: the covering map $\pi$ has degree $n/m$ along the fibers of $f$ and degree $m$ along the base $\P^1$. So 
$\pi^{-1}$ of a generic abelian surface fiber splits into $m$ disconnected
components, each an abelian surface. In other words, $m$ can be recovered from the topology of $Y$ plus the fibration $f$.
So if we know that the fibration $f$ is unique, it follows 
that $m$ can be used to distinguish threefolds.

To recover $f$ for the generic member $Y$ in each of our families, we 
assume that $\widetilde Y$ is the fiber product of two 
rational elliptic surfaces $S_1,S_2$, and that there exists a point of $\P^1$ such that
the two elliptic fibers $E_1,E_2$ over it are not isogenous. This can be arranged
since by moving in the moduli space of $Y$ we can vary the $j$-function continuously.
Then the generic fiber $A = E_1 \times E_2$ of 
$\widetilde f := f \circ\pi$ 
is the product of two non isogenous elliptic curves. The only line
bundles on such an $A$ are products of pullbacks from the two components. Any map $A \to \P^1$ is given by such a line bundle of self-intersection $0$, hence the line bundle must be a pullback from a single $E_i$, and the map must factor through that $E_i$. Therefore any map ${\widetilde f}^\prime : \widetilde Y \to \P^1$ must factor through an elliptic fibration on one of the rational elliptic surfaces $S_i$. But the elliptic fibration on $S_i$ is unique, and is given by $E_i$ in the anticanonical system: the connected component $C$ of the general fiber of any other fibration on $S_i$ has positive intersection number with $E_i$, so by adjunction it has to be rational rather than elliptic.
This proves that the fibration $f$ is unique. It follows that Schoen quotients with distinct invariants $m$ are non isomorphic as algebraic varieties. Since each family of Schoen quotients dominates its complex structure moduli space, it also follows that Schoen quotients with distinct invariants $m$ are not deformation equivalent.
%
%
%
\renewcommand{\thesection}{\Alph{section}}
\renewcommand{\theequation}{\Alph{section}.\arabic{equation}}
\setcounter{section}{0}
\section{Jacobi theta functions and their properties}\label{thetafctns}
We use the following  functions of $q=e^{2\pi i\tau}$, 
$\tau\in\H$, $\H=\left\{\tau\in\C\mid \Im(\tau)>0\right\}$ and $y=e^{2\pi i z},\,z\in\C$,
\beqno
\theta_1(\tau,z)  =   -\theta_{11}(\tau,z)
& := & 
i\sum_{n=-\infty}^\infty (-1)^n q^{{1\over2}(n-{1\over2})^2} y^{n-{1\over2}}\e 
& =&
i q^{{1\over8}} y^{-{1\over2}} \prod_{n=1}^\infty 
(1-q^n)(1-q^{n-1}y)(1-q^{n}y^{-1}),\\[10pt]\ds
\theta_2(\tau,z)  =\;\;\,  \theta_{10}(\tau,z) 
& := &
\sum_{n=-\infty}^\infty q^{{1\over2}(n-{1\over2})^2} y^{n-{1\over2}}\e
& =&
q^{{1\over8}} y^{-{1\over2}} \prod_{n=1}^\infty 
(1-q^n)(1+q^{n-1}y)(1+q^{n}y^{-1}),\\[10pt]\ds
\theta_3(\tau,z)  = \;\;\,  \theta_{00}(\tau,z) 
& := &
\sum_{n=-\infty}^\infty q^{{n^2\over2}} y^{n}\e 
& = &
\prod_{n=1}^\infty (1-q^n)(1+q^{n-{1\over2}}y)(1+q^{n-{1\over2}}y^{-1}),
\\[10pt]\ds
\theta_4(\tau,z)  = \;\;\,  \theta_{01}(\tau,z) 
& := &
\sum_{n=-\infty}^\infty (-1)^n q^{{n^2\over2}} y^{n}\e 
& =&
\prod_{n=1}^\infty (1-q^n)(1-q^{n-{1\over2}}y)(1-q^{n-{1\over2}}y^{-1}).  
\eeqno
The functions $\theta_1,\,\theta_2,\,\theta_3,\,\theta_4$ are commonly known as
Jacobi theta functions.
We frequently denote $\theta_k(\tau):=\theta_k(\tau,0)$ or even 
$\theta_k:=\theta_k(\tau,0)$, so in particular since $\theta_1(\tau,z)$ is an odd
function in $z$, 
$\theta_1=0$.

The following transformation laws are obtained directly
from the definition or by Poisson resummation:
\beqno\label{anh_th_trafo}
\begin{array}{|l||c|c|c|c|}
\hline\vphantom{\displaystyle\int}
\mbox{Operation} 
& \theta_1(\tau)  & \theta_2(\tau) & \theta_3(\tau) & \theta_4(\tau) \\
\hline\hline\hline
\vphantom{\displaystyle\int}
\tau \mapsto \tau + 1 
& e^{ {2\pi i\over 8} } \theta_1(\tau,z)& e^{ {2\pi i\over 8} } \theta_2(\tau,z)
& \theta_4(\tau,z)&  \theta_3(\tau,z) \\
\hline\vphantom{\displaystyle\int}
\tau \mapsto -{1\over\tau} , 
& (-i)(-i\tau)^{ {1\over 2} }e^{ {\pi i z^2\over \tau} } \cdot
& (-i\tau)^{ {1\over 2} }e^{ {\pi i z^2\over \tau} } \cdot
& (-i\tau)^{ {1\over 2} }e^{ {\pi i z^2\over \tau} }\cdot
& (-i\tau)^{ {1\over 2} }e^{ {\pi i z^2\over \tau} } \cdot \\
\quad z \mapsto {z\over\tau}
& \quad\quad \cdot\theta_1(\tau,z)
& \quad\quad \cdot\theta_4(\tau,z)
& \quad\quad \cdot\theta_3(\tau,z)
& \quad\quad \cdot\theta_2(\tau,z) \\[0.2em]
\hline
\end{array}\nonumber\\
\eeqno
We also use the Dedekind eta function
$$
\eta=\eta(\tau):=q^{1/24} \prod_{n=1}^\infty \left( 1-q^n  \right).
$$
Under modular transformations, it obeys
$$
\eta(\tau+1)=e^{2\pi i/24} \eta(\tau), \quad
\eta({\textstyle-{1\over\tau}})=(-i\tau)^{{1\over2}} \eta(\tau).
$$
By using the Jacobi triple identity one can prove the following product
formulas:
\beqn{vanish}
\theta_2(\tau)\theta_3(\tau)\theta_4(\tau) &=& 2 \eta(\tau)^3\nonumber\\
\theta_2(\tau)^4 - \theta_3(\tau)^4 + \theta_4(\tau)^4 &=& 0 . 
\eeqn
\section{Representations of the free fermion algebra}\label{ffr}
A single free fermion $\psi$ can have one of four different spin structures,
each characterized by two binaries
$\alpha,\,\beta\in\{0,1\}$.
The fermion $\psi$
is said to belong to the NS (Neveu-Schwarz) sector if $\alpha=0$, 
where it has half integer (Fourier) modes on expansion with respect
to the parameter $x\in\C^\ast$ of the field $\psi$, and otherwise
it belongs to the R (Ramond)
sector, where it has integer modes. The modes obey
$$
\left\{\psi_a,\psi_b\right\} = \delta_{a+b,0} \;\mbox{ for }\;
a,b\in\left\{ \begin{array}{ll}
\Z+\inv{2}&\mb{(NS)}\\
\Z&\mb{(R)}
\end{array}\right.
$$
and thus act as creation or annihilation operators. 
These modes together with $1$ (that is, a central element which in each 
representation is normalized to act as identity operator)
form a vector space basis of the so-called free fermion algebra.

Let $\m H_0,\,\m H_1$
denote the irreducible Fock space representations of the free fermion
algebra in the NS and the R sector, respectively, enlarged by
$(-1)^F$ with $F$ the worldsheet fermion number, i.e.\
such that $(-1)^F$ is a non-trivial involution which anticommutes
with all $\psi_a$.
Each state in $\m H_0,\,\m H_1$
is obtained by acting with pairwise distinct fermionic creation operators
on a ground state and thereby increasing the conformal dimension 
by half integer (NS) or integer (R) steps. In the NS sector, 
ground states of this Fock space representation of the free 
fermion algebra have conformal dimension $h=0$, whereas in the R sector, 
they have conformal dimension $h={1\over16}$. 
In fact, $\m H_0$ has a unique ground state (up to scalar multiples)
$|0\rangle$, the vacuum, 
whereas $\m H_1$ possesses a two 
dimensional space of such ground states. 
The vacuum $|0\rangle$ is a worldsheet boson, i.e. 
$(-1)^F|0\rangle=|0\rangle$,
and in $\m H_1$ we choose a basis $|\pm\rangle$ of ground states
such that $|+\rangle$ is a worldsheet boson and $|-\rangle$ is
a worldsheet fermion, i.e. $(-1)^F|\pm\rangle=\pm|\pm\rangle$. 
The decomposition of $\m H_0,\,\m H_1$
into worldsheet bosons and worldsheet fermions,
$$
\m H_{\alpha}\cong\m H_{\alpha}^+\oplus\m H_{\alpha}^-,\quad
\alpha\in\{0,1\}
$$ 
agrees with the decomposition into irreducible representations of 
the Virasoro algebra at central charge $c=\inv2$ which arises
from the universal enveloping algebra of the free fermion algebra
in either sector.
We set
\beq{theta}
\mbox{for } \alpha,\beta\in\{0,1\}:\quad
Z\!\!\left[\!\!\begin{array}{c}\alpha\\\beta\end{array}\!\!\right]
:=\sqrt{\vartheta_{\alpha,\beta}\over\eta},
\eeq
where the $\vartheta_{\alpha,\beta}$ denote the Jacobi theta functions and $\eta$
the Dedekind eta function listed in 
Appendix \ref{thetafctns}. The square root makes sense in terms of the infinite
product representations of the $\vartheta_{\alpha,\beta}$ also given there.
Then with $q=e^{2\pi i\tau}$ and $\tau$ as before, the above
discussion together with the explicit product formulas given in 
Appendix \ref{thetafctns} shows
$$
\begin{array}{rclrcl}
Z\!\!\left[\!\!\begin{array}{c}0\\0\end{array}\!\!\right]\!\!(\tau)\!\!
&=&\tr[\m H_0] \left[q^{L_0-1/48}\right],\quad
&Z\!\!\left[\!\!\begin{array}{c}0\\1\end{array}\!\!\right]\!\!(\tau)\!\!
&=&\tr[\m H_0] \left[(-1)^{F}\, q^{L_0-1/48}\right],\\[1em]
Z\!\!\left[\!\!\begin{array}{c}1\\0\end{array}\!\!\right]\!\!(\tau)\!\!
&=&{\textstyle{1\over\sqrt2}} \tr[\m H_1] \left[q^{L_0-1/48}\right],\quad
&Z\!\!\left[\!\!\begin{array}{c}1\\1\end{array}\!\!\right]\!\!(\tau)\!\!
&=&{\textstyle{1\over\sqrt2}} 
\tr[\m H_1] \left[(-1)^{F}\, q^{L_0-1/48}\right] = 0.
\end{array}
$$
The insertion of $(-1)^{F}$ in the traces to obtain
$Z\!\!\left[\!\!\begin{array}{c}\alpha\\1\end{array}\!\!\right]$ 
from
$Z\!\!\left[\!\!\begin{array}{c}\alpha\\0\end{array}\!\!\right]$ 
corresponds in Hamiltonian
language to changing the spin structure of the fermion in the imaginary
time direction. This means that 
$Z\!\!\left[\!\!\begin{array}{c}\alpha\\\beta\end{array}\!\!\right]$ 
gives the 
contribution to the partition function of a free fermion with spin 
structure specified by 
$\alpha,\,\beta\in\{0,1\}$.

The factors of ${1\over\sqrt2}$ in the traces for the R-sector 
yield $2\cdot{1\over\sqrt2}=\sqrt2$ as coefficient of the leading term $q^{{1\over24}}$
in $Z\!\!\left[\!\!\begin{array}{c}1\\0\end{array}\!\!\right]$,
accounting for the contributions from the space generated by $|+\rangle$ and $|-\rangle$. 
We obtain integer coefficients as soon as
we consider pairs of fermions with coupled spin structures in space
direction, which is necessary anyway in order to get pairwise local fields of a
well-defined CFT. Given a collection of free fermions,
for a tensor product between the
R-sectors of the  
$j^{\mbox{\tiny{}th}}$ and the $j^\prime{}^{\,\mbox{\tiny th}}$ free 
fermion, 
with coupled spin structures, 
$\m H_1^j\otimes\m H_1^{j^\prime}$ 
splits into two isomorphic representations of 
the free fermion algebras 
generated by $\psi_a^j,\,\psi_a^{j^\prime}$ with $a\in\Z$ enlarged by the 
total worldsheet fermion number operator $(-1)^{F_j+F_{j^\prime}}$,
one with ground states
$$
|+\rangle\otimes |+\rangle + |-\rangle\otimes |-\rangle, \quad
|+\rangle\otimes|-\rangle + |-\rangle\otimes |+\rangle,
$$
the other with ground states 
$$
|+\rangle\otimes |+\rangle - |-\rangle\otimes |-\rangle, \quad
|+\rangle\otimes|-\rangle - |-\rangle\otimes |+\rangle,
$$
respectively.
Let pr denote the projection onto one of these two representations. 
Using $\left(Z\!\!\left[\!\!\begin{array}{c}1\\0\end{array}\!\!\right]\right)^2$
as above then gives the trace over pr$\left(\m H_1^j\otimes\m H_1^{j^\prime}\right)$, as
the coefficients ${\sqrt2}$ conspire correctly to count a
two-dimensional space of ground states.
Since $\left(Z\!\!\left[\!\!\begin{array}{c}1\\1\end{array}\!\!\right]\right)^2=0$,
one of the generators is correctly counted as boson, the other as fermion. 
Summarizing, if $\psi^j$ and $\psi^{j^\prime}$ have coupled spin structures
$(\alpha_j,\beta_j)=(\alpha_{j^\prime}, \beta_{j^\prime})$ and pr is extended
trivially to $\m H_0^j\otimes\m H_0^{j^\prime}$, then
\beq{traceformula}
Z\!\!\left[\!\!\begin{array}{c}\alpha_j\\\beta_j\end{array}\!\!\right]
\cdot Z\!\!\left[\!\!\begin{array}{c}\alpha_{j^\prime}\\\beta_{j^\prime}\end{array}\!\!\right]
= \tr[{\rm pr}\left(\m H_{\alpha_j}^j\otimes\m H_{\alpha_{j^\prime}}^{j^\prime}\right)]
\left[ (-1)^{\beta_j F_j+\beta_{j^\prime} F_{j^\prime}} q^{L_0-1/48} \right].
\eeq
%
%
%
\def\polhk#1{\setbox0=\hbox{#1}{\ooalign{\hidewidth
  \lower1.5ex\hbox{`}\hidewidth\crcr\unhbox0}}}
\providecommand{\bysame}{\leavevmode\hbox to3em{\hrulefill}\thinspace}

\end{document}